\documentclass[12pt]{article}

\usepackage{amsfonts}
\usepackage{amssymb}
\usepackage{amsmath}

\usepackage{titling}
\usepackage{bibunits}
\usepackage{booktabs}
\usepackage{natbib}
\usepackage{url}
\usepackage{fontawesome}
\usepackage[breaklinks=true,bookmarksopen=true,colorlinks=true,citecolor=blue]{hyperref}

\usepackage{setspace}

\newcommand{\blind}{0}

\newcommand{\aidisclosure}{The authors declare the use of generative AI. The research question, the modeling framework, and the initial roadmap of results are the authors' own. Under the authors' direction and review, AI agents (Claude Opus 5 and Claude Fable 5 by Anthropic) helped build the numerical exercise and prove the results reported here, while various separate agent sessions (Gemini 3.1 Pro by Google and GPT 5.6 Soltice by OpenAI through \href{https://coarse.ink/}{https://coarse.ink/}) independently verified the mathematical correctness of all derivations. The authors have adjudicated all AI-assisted content and take full responsibility for the accuracy, integrity, and originality of the final manuscript.}

\newtheorem{theorem}{Theorem}
\newtheorem{lemma}{Lemma}[section]
\newtheorem{proposition}{Proposition}

\newtheorem{assumption}{Assumption}

\newenvironment{proof}[1][Proof]{\noindent\textbf{#1.} }{\ \rule{0.5em}{0.5em}}

\def\spacingset#1{\renewcommand{\baselinestretch}{#1}\small\normalsize}

\numberwithin{equation}{section}

\begin{document}
\begin{bibunit}[jpe]

\spacingset{1}

\if1\blind
{
  \title{\bf Rate-Agnostic Wald Inference for Dyadic Regressions\thanks{\aidisclosure}}
  \author{Anonymous}
  \date{\today}
  \maketitle
} \fi

\if0\blind
{
  \title{\bf Rate-Agnostic Wald Inference for Dyadic Regressions\thanks{\aidisclosure}}
  \author{Benjamin O. Harrison\thanks{Department of Economics, Emory University,
  Rich Building 306, 1602 Fishburne Dr., Atlanta, GA 30322-2240, USA.
  \faEnvelopeO: \href{mailto:boharri@emory.edu}{boharri@emory.edu}.}
  \and David T. Jacho-Ch\'{a}vez\thanks{Corresponding author. Department of Economics,
  Emory University, Rich Building 306, 1602 Fishburne Dr., Atlanta, GA 30322-2240, USA.
  \faEnvelopeO: \href{mailto:djachocha@emory.edu}{djachocha@emory.edu}.}}
  \date{\today}
  \maketitle
} \fi

\begin{abstract}
\normalsize
\noindent This paper develops Wald inference for least-squares estimation of linear regression models on dyadic data, accommodating configurations where multiple observations share the same pair of units (e.g., directed flows, multilayer networks, and dyadic panels). We establish that the dyadic-robust Wald statistic is asymptotically $\chi^2_q$ for an arbitrary nonrandom sequence of full-rank restrictions, under a single condition on the accumulation of dependence. Throughout, no convergence rate is assumed or estimated, permitting the condition number of the score's variance matrix to diverge. We further propose a delete-one-unit jackknife alternative that is positive semidefinite by construction. This jackknife statistic attains the same asymptotic limit under one additional condition on dyad multiplicity and remains asymptotically conservative when that condition fails. A supplement contains all proofs, Monte Carlo experiments featuring estimated coefficients that converge at heterogeneous rates, and an empirical gravity application to bilateral trade.
\end{abstract}

\noindent%
\noindent%

\spacingset{1.15}

\section{Introduction}
\label{sec:intro}

In this paper we provide asymptotically valid Wald inference for least squares in linear regression models where each observation is attached to a \emph{pair} of units, a dyad. Several observations may be attached to the \emph{same} pair. Examples of this setting include the two directed flows between two countries, the layers of a multilayer network, or the dates of a dyadic panel. Our framework encompasses all these settings at once by allowing the map from observations to dyads to be non-injective. Dyadic observations may be arbitrarily dependent when their underlying dyads share at least one common unit; however, observations with disjoint dyads are assumed to be independent. In particular, this accommodates reciprocity by leaving the correlation between the two orientations of a given pair unrestricted.

Inference in this setting is routinely conducted with the dyadic-robust variance
estimator of \citet{fafchamps2007formation} and \citet{aronow2015cluster}, surveyed in \citet{graham2020network}. \citet{tabordmeehan2019inference} provides the corresponding limit theory for dyadic least squares inference, establishing conditions under which the dyadic-robust $t$-statistic is asymptotically standard normal. However, the theory relies on a common convergence rate across the estimator's components. By contrast, \citet[p.~277]{hansen2019asymptotic} develop a rate-free inference framework whose Theorem~9 permits different linear combinations of an estimator to converge at different rates, but their reliance on mutually independent clusters precludes its application to the dyadic configurations studied here.

Our contribution bridges these two frameworks. We retain the dyadic dependence structure formalized by \citet{tabordmeehan2019inference}, but dispense with the common-rate requirement by adapting the rate-free logic of \citet{hansen2019asymptotic} to a setting without mutually independent clusters. Theorem \ref{thm:main} establishes asymptotically valid Wald inference without imposing a common scalar convergence rate. The Wald statistic is asymptotically $\chi^2_q$ for an \emph{arbitrary} nonrandom sequence of full-rank restriction matrices, subject to a single condition on the accumulation of dependence. This condition assumes no scalar convergence rate and permits different linear combinations of the estimator to converge at heterogeneous rates. This matters in practice, since heterogeneous rates are the norm on dyadic designs often used in the trade literature: a coefficient on an exporter attribute and one on a trading-pair attribute accumulate dependence at different speeds, and a researcher running a gravity regression cannot be asked to know, or estimate, either rate.

While Theorem \ref{thm:main} establishes the asymptotic properties of the dyadic-robust variance estimator, our second contribution addresses its finite-sample behavior. Because its inner matrix is a difference of two positive semidefinite matrices, the estimator is not guaranteed to be positive semidefinite at a given sample size, which can render the resulting Wald statistic undefined. We therefore study a jackknife alternative that deletes one \emph{unit}, and with it every observation that unit touches, at a time. The construction is that of
\citet{snijders1999nonparametric}, rescaled to consistency for a dyadic M-estimator in \citet{graham2020network}; see also \citet{bell2002bias} and
\citet{lin2020theoretical} for similar constructions in other settings. It is a sum of outer products, and is therefore positive semidefinite. Because each observation belongs to exactly two deletion groups, however, the classical argument over disjoint independent clusters is unavailable. Theorem \ref{thm:jack} shows that the jackknife Wald statistic attains the same $\chi^2_q$ limit under one further condition on dyad multiplicity, and Proposition \ref{prop:onesided} shows that without that condition the failure is one-sided, in the sense that the test never over-rejects asymptotically.

Finally, we provide three technical contributions that might be of independent interest. First, we establish a central limit theorem for the least squares score standardized by its exact variance rather than by a rate; this accommodates triangular arrays where the tested direction changes with the sample size. Second, we prove that the dyadic-robust variance estimator is consistent relative to the true variance uniformly over the restrictions tested, permitting its linear combinations to converge at heterogeneous rates. Third, we derive an exact finite-sample identity for the effect of deleting one unit on the least squares estimator. Because each observation belongs to exactly two deletion groups, this identity reveals that the first-order terms sum exactly to zero, rendering the difference between the jackknife mean and the full-sample estimator purely of second order.

The remainder of the paper is organized as follows. Section \ref{sec:model} presents the model, the dyadic configuration notation, and the two variance estimators. Section \ref{sec:assumptions} states the main results and compares our assumptions with the existing literature, and Section \ref{sec:conclusion} concludes. All proofs are deferred to the supplemental materials. The supplement also provides three worked-out illustrations and a numerical exercise evaluating a gravity equation across three dyadic configurations: directed Erd\H{o}s--R\'{e}nyi configurations of varying density, multilayer configurations featuring heterogeneous convergence rates, and an empirical application to the bilateral trade data of \citet{santossilva2006log}.

\section{Model and Estimators}
\label{sec:model}

Units are indexed $g=1,\dots,G$ and all limits are taken as $G\to\infty$. Observations are
indexed $n=1,\dots,N$ with $N=N_G$. Each carries a support $\psi(n)\subset
\{1,\dots,G\}$ with $|\psi(n)|=2$, called its dyad, and $\psi$ is \emph{not} assumed
injective, so several observations may carry the same dyad. Write $M_g :=
\#\{n : g\in\psi(n)\}$ for the number of observations touching unit $g$ (observations, not
distinct partners), $\mathcal{M}^{H} := \max_g M_g$ and $\mathcal{M}^{L} := \min_g M_g$. Dually, for a dyad $d$, $L_d := \#\{n : \psi(n)=d\}$, $\mathcal{L} := \max_d L_d$, and
$\mathcal{J}_N := \sum_d L_d^2$, the number of ordered pairs of observations carrying the
same dyad. Since each observation contributes to exactly two of the $M_g$ and $\sum_d L_d
= N$,
\begin{equation}\label{eq:degrees}
\frac{\mathcal{M}^{L}G}{2} \;\le\; N \;\le\; \frac{\mathcal{M}^{H}G}{2},
\qquad\qquad
N \;\le\; \mathcal{J}_N \;\le\; \mathcal{L}N, \qquad \mathcal{L}\le\mathcal{M}^{H}\le N .
\end{equation}
Let $\mathbf{1}_{nm} := \mathbf{1}\{\psi(n)\cap\psi(m)\ne\emptyset\}$ and $\mathcal{A} :=
\{(n,m) : \mathbf{1}_{nm}=1\}$. For matrices, $\|\cdot\|$ and $\|\cdot\|_F$ denote the
spectral and Frobenius norms, $\lambda_{\min}$ and $\lambda_{\max}$ extreme eigenvalues,
and for symmetric positive definite $A$, $A^{1/2}$ is the symmetric square root with
$A^{-1/2} := (A^{1/2})^{-1}$. Finally, throughout, $C := C_x^2C_u^2$ denotes the constant
built from Assumption \ref{as:bound} below.

The linear model is $y_n = \boldsymbol{\beta}'\mathbf{x}_n + u_n$, with
$\mathbf{x}_n\in\mathbb{R}^{K}$ and $K$ fixed, and we estimate it by ordinary least
squares. Let $\mathcal{E}_N$ denote the event that $\sum_n \mathbf{x}_n\mathbf{x}_n'$ is
invertible, and write $\mathcal{E}_N^{c}$ for its complement. On $\mathcal{E}_N$ we set
$\widehat{\boldsymbol{\beta}} := (\sum_n\mathbf{x}_n\mathbf{x}_n')^{-1}
\sum_n\mathbf{x}_ny_n$ and $\widehat{u}_n := y_n -
\widehat{\boldsymbol{\beta}}'\mathbf{x}_n$, while on $\mathcal{E}_N^{c}$ we set
$\widehat{\boldsymbol{\beta}} := \mathbf{0}$ and $\widehat{u}_n := y_n$. Define the score,
its variance, and the design matrices $S_N := \sum_n\mathbf{x}_nu_n$, $\Omega_N :=
\operatorname{Var}(S_N)$, $Q_N := N^{-1}\sum_n\mathbb{E}[\mathbf{x}_n\mathbf{x}_n']$, and
$\widehat{Q}_N := N^{-1}\sum_n\mathbf{x}_n\mathbf{x}_n'$. The \emph{dyadic-robust variance
estimator} of $\widehat{\boldsymbol{\beta}}$ is built
from the feasible and infeasible meats
\begin{equation}\label{eq:meat}
\widehat{\Omega}_N \;:=\; \sum_{n=1}^{N}\sum_{m=1}^{N} \mathbf{1}_{nm}\,
\widehat{u}_n\widehat{u}_m\,\mathbf{x}_n\mathbf{x}_m', \qquad
\widetilde{\Omega}_N \;:=\; \sum_{n=1}^{N}\sum_{m=1}^{N} \mathbf{1}_{nm}\,
u_nu_m\,\mathbf{x}_n\mathbf{x}_m',
\end{equation}
through the sandwich
\begin{equation}\label{eq:vhat}
\widehat{V}_N \;:=\; \frac{1}{N^{2}}\,\widehat{Q}_N^{-1}\widehat{\Omega}_N
\widehat{Q}_N^{-1}, \qquad\qquad V_N \;:=\; \frac{1}{N^{2}}\,Q_N^{-1}\Omega_NQ_N^{-1} .
\end{equation}
Since $\widehat{Q}_N$ is singular on $\mathcal{E}_N^{c}$, we set $\widehat{V}_N:=\mathbf{0}$ there. Alternatively, one can use a Moore--Penrose inverse \citep[see, e.g.,][p.~144]{graham2020network}; either approach is valid because
$\mathbb{P}(\mathcal{E}_N)\to1$ under Assumptions \ref{as:dep}--\ref{as:config}.
Here $\widehat{V}_N$ estimates $V_N$, which standardizes
$\widehat{\boldsymbol{\beta}}-\boldsymbol{\beta}$ in Theorem \ref{thm:main} below.
The second variance estimator is the aforementioned
delete-one-unit jackknife, i.e., $\widehat{V}^{J}$. Write $A := \sum_n\mathbf{x}_n\mathbf{x}_n'$,
$\mathcal{N}_g := \{n : g\in\psi(n)\}$, $A_g := \sum_{n\in\mathcal{N}_g}
\mathbf{x}_n\mathbf{x}_n'$ and $A_{(g)} := A - A_g$; on the event $\mathcal{E}^{J}_N$ that
every $A_{(g)}$ is invertible, let $\widehat{\boldsymbol{\beta}}_{(g)} :=
A_{(g)}^{-1}\sum_{n\notin\mathcal{N}_g}\mathbf{x}_ny_n$,
$\bar{\boldsymbol{\beta}} := G^{-1}\sum_g\widehat{\boldsymbol{\beta}}_{(g)}$, and
\begin{equation}\label{eq:vjack}
\widehat{V}^{J} \;:=\; \frac{G-1}{G}\sum_{g=1}^{G}
\big(\widehat{\boldsymbol{\beta}}_{(g)} - \bar{\boldsymbol{\beta}}\big)
\big(\widehat{\boldsymbol{\beta}}_{(g)} - \bar{\boldsymbol{\beta}}\big)' ,
\end{equation}
with $\widehat{V}^{J} := \mathbf{0}$ on $(\mathcal{E}^{J}_N)^{c}$ as well. Deleting unit $g$
removes all $M_g$ observations touching it at once. Two features of \eqref{eq:vjack}
matter below. First, $\widehat{V}^{J}$ is positive semidefinite by construction, which
$\widehat{\Omega}_N$ is not (see, e.g., \citealp{mackinnon2023cluster}, for further
details). Second, each observation belongs to exactly two of the $G$ deletion groups, so
the classical disjoint-cluster argument \citep{cameron2011robust, hansen2019asymptotic}
cannot be used in the mathematical proofs.

\section{Assumptions and Main Results}
\label{sec:assumptions}

We consider the following regularity conditions.

\begin{assumption}[Dyadic dependence]\label{as:dep}
For any two disjoint
index sets $S_1,S_2\subset\{1,\dots,N\}$, the families
$\{(\mathbf{x}_n,u_n)\}_{n\in S_1}$ and $\{(\mathbf{x}_m,u_m)\}_{m\in S_2}$ are
independent whenever $\psi(n)\cap\psi(m)=\emptyset$ for every $n\in S_1$ and $m\in S_2$.
\end{assumption}

\begin{assumption}[Bounded support]\label{as:bound}
There are finite constants $C_x,C_u$ with $\|\mathbf{x}_n\|\le C_x$ and
$|u_n|\le C_u$ almost surely for all $n$ and $N$.
\end{assumption}

\begin{assumption}[Identification]\label{as:ident}
$\mathbb{E}[\mathbf{x}_nu_n]=\mathbf{0}$ for all $n$, and
$\lambda_{\min}(Q_N)\ge c_Q>0$ for all $N$.
\end{assumption}

\begin{assumption}[Variance accumulation]\label{as:acc}
$\Omega_N$ is positive definite and\newline
$\delta_N := N(\mathcal{M}^{H})^{3}\big/\lambda_{\min}(\Omega_N)^{2}\longrightarrow 0$.
\end{assumption}

\begin{assumption}[Configuration]\label{as:config}
$K$ is fixed and $\mathcal{M}^{L}\ge1$ as $G\to\infty$.
\end{assumption}

Assumption \ref{as:dep} is the dyadic analogue of independent clustering, and it is
a weaker version of Assumption 2.1 in \citet[p.~672]{tabordmeehan2019inference} in that the support map need not be injective and the observations need not be identically distributed. Note that only the independence of families
indexed by node-disjoint dyads is ever used below. This assumption therefore imposes restrictions only on observations that do not share a unit. Consequently, any two observations sharing a common unit, whether they represent two orientations of a single pair, distinct network layers, or different time periods, may be arbitrarily dependent. In particular, reciprocity and cross-layer correlation are left unrestricted. Assumption \ref{as:dep} also generalizes the sampling scheme of \citet[Section 2, p.~269]{hansen2019asymptotic} because it does not require clusters to be disjoint. In their framework, clusters are mutually independent by construction. In contrast, the natural clusters $\{n:g\in\psi(n)\}$ in our setting overlap, with each observation belonging to exactly two of them. Their scheme can be recovered by placing each mutually independent cluster on its own isolated pair of units. However, Assumption \ref{as:dep} does rule out dependence across node-disjoint dyads. Consequently, any shock common to all dyads must be absorbed into the mean, either by conditioning on its realization or by including it as a nonrandom shifter, rather than subsumed in the error term.

Assumption \ref{as:ident} is standard in the linear regression literature; see, for example, the conditions in \citet[p.~672]{tabordmeehan2019inference}. It only requires the regressors and the error term to be uncorrelated, rather than imposing strict mean independence $\mathbb{E}[u_n\mid\mathbf{x}_n]=0$. Thus, $\boldsymbol{\beta}$ is defined as a projection parameter rather than a conditional-mean coefficient. It is important to note, however, that Assumption \ref{as:ident} imposes orthogonality observation by observation. When the data are not identically distributed, this condition is stronger than the aggregate restriction $\sum_n\mathbb{E}[\mathbf{x}_nu_n]=\mathbf{0}$ implied by a pooled population projection. These individual restrictions are precisely what Lemma \ref{lem:struct}(a) uses to ensure the raw cross-products coincide with their covariances and that node-disjoint cross-products vanish, thereby identifying $\mathbb{E}[\widetilde{\Omega}_N]$ with $\Omega_N$. The estimand is therefore a common coefficient satisfying pair-specific orthogonality, and it is imposed as such in the empirical application in Appendix \ref{sec:empirical}.

Similarly, Assumption \ref{as:bound} is the bounded support condition of Assumption 3.1 in
\citet[p.~675]{tabordmeehan2019inference}. It is stronger than the moment conditions
under which cluster-robust limit theory is usually developed, and it is used here to bound
the summands of $\Omega_N$ uniformly, which with the pair count of Lemma
\ref{lem:struct}(b) gives $\lambda_{\max}(\Omega_N)=O(N\mathcal{M}^{H})$. That cap is what permits the condition number of $\Omega_N$ to
diverge in Theorem \ref{thm:main} below. When clusters are disjoint the same control
follows from bounded moments alone, and it does not when they overlap. Notice that Assumption
\ref{as:bound} restricts the support of $(\mathbf{x}_n,u_n)$ and nothing else about its
distribution, so it imposes neither homoskedasticity nor any parametric form. Assumption
\ref{as:config} requires that no unit be isolated and that $K$ stay fixed as $G$ grows.
That is, it rules out node fixed effects and slope vectors whose dimension grows with
$G$; see Illustration
3 in \ref{sec:smb}.

Assumption \ref{as:acc} is more substantive. It replaces Assumption 2.6 in
\citet[p.~675]{tabordmeehan2019inference}, which requires $(NG^{r})^{-1}\Omega_N$ to
converge to a positive-definite limit for some exponent $r\in[0,1]$. No such rate is
imposed here,
and no normalization of $\Omega_N$ is required to converge. What is asked instead is
positive definiteness at each $N$ together with $\delta_N\to0$. Dispensing with the limit
is what lets Theorem \ref{thm:main} below test joint restrictions whose directions
accumulate dependence at genuinely different orders, and Appendix \ref{sec:mc2} reports a
design in which they do. Its second part, namely
$\delta_N\to0$, is similar to the eigenvalue condition in Theorem 9 of
\citet[p.~277]{hansen2019asymptotic}, where the smallest eigenvalue of a normalized
cluster variance is bounded away from zero, but stronger on dense configurations, in the sense that
$\lambda_{\min}(\Omega_N)$ must outgrow $\{N(\mathcal{M}^{H})^{3}\}^{1/2}$, a bound that
exceeds the order $N$ asked for there exactly when $(\mathcal{M}^{H})^{3}>N$. The numerator
$N(\mathcal{M}^{H})^{3}$ counts the overlapping index quadruples that control the
variance of the meat (see Lemma \ref{lem:struct}(c)), a count that is of order $N$ when
clusters are bounded and independent, and grows with the maximum degree once they
overlap.

Notice that Assumption \ref{as:acc} asks nothing about the rate at which
$\widehat{\boldsymbol{\beta}}$ converges, and no such rate is either assumed or estimated
anywhere below. The matrix $\Omega_N$ need not converge under any normalization, and its
condition number $\kappa(\Omega_N):=\lambda_{\max}(\Omega_N)/\lambda_{\min}(\Omega_N)$ is
free to diverge, so that different linear combinations of $\widehat{\boldsymbol{\beta}}$
may converge at different rates, as the multilayer designs of Illustration 2 in
\ref{sec:smb} exhibit. The spread of those rates is nevertheless
controlled, since $\delta_N\to0$ implies
$\kappa(\Omega_N)=o\big((N/\mathcal{M}^{H})^{1/2}\big)$. That is, \emph{rate-agnostic}
refers to what has to be known in order to conduct inference, and not to how far apart
the underlying rates may lie. If the design is nonrandom and the errors satisfy
$\operatorname{Var}(\mathbf{u})\succeq\underline{\sigma}^{2}I_N$ for some
$\underline{\sigma}^{2}>0$, then $\Omega_N\succeq\underline{\sigma}^{2}\sum_n
\mathbf{x}_n\mathbf{x}_n'$ and hence $\lambda_{\min}(\Omega_N)\ge
\underline{\sigma}^{2}c_QN$ by Assumption \ref{as:ident}, so that Assumption
\ref{as:acc} holds automatically whenever
$(\mathcal{M}^{H})^{3}=o(N)$, a requirement on the configuration alone which by
\eqref{eq:degrees} entails $\mathcal{M}^{H}=o(\sqrt{G})$, and is therefore available on
sparse designs. On denser
ones it is a joint restriction on the configuration, the design and the error structure. The following theorem establishes the consistency of the sandwich estimator in
\eqref{eq:vhat} and the limiting distribution of the implied Wald statistic.

\begin{theorem}\label{thm:main}
Let Assumptions \ref{as:dep}--\ref{as:config} hold, and let $\{R_N\}$ be an arbitrary
sequence of nonrandom $K\times q$ matrices of full column rank $q\le K$. Write
$\mathcal{V}_N := R_N'V_NR_N$ and $\widehat{\mathcal{V}}_N := R_N'\widehat{V}_NR_N$. Then, as $G\to\infty$:
(a) $\mathcal{V}_N^{-1/2}R_N'(\widehat{\boldsymbol{\beta}}-\boldsymbol{\beta})
\longrightarrow_d N(\mathbf{0},I_q)$;
(b) $\mathcal{V}_N^{-1/2}\widehat{\mathcal{V}}_N\mathcal{V}_N^{-1/2}\longrightarrow_p I_q$;
(c) $\mathbb{P}(\widehat{\mathcal{V}}_N\succ0)\to1$, and on that event
$\widehat{\mathcal{V}}_N^{-1/2}R_N'(\widehat{\boldsymbol{\beta}}-\boldsymbol{\beta})
\longrightarrow_d N(\mathbf{0},I_q)$;
(d) under $H_0: R_N'\boldsymbol{\beta}=\mathbf{r}$, the Wald statistic
\[
W_N := \big(R_N'\widehat{\boldsymbol{\beta}}-\mathbf{r}\big)'\widehat{\mathcal{V}}_N^{-1}
\big(R_N'\widehat{\boldsymbol{\beta}}-\mathbf{r}\big) \;\longrightarrow_d\; \chi^2_q .
\]
\end{theorem}

Note that parts (a) and (b) are statements on $\mathcal{E}_N$, whose probability tends to
one, and that part (b) delivers $\mathbb{P}(\widehat{\mathcal{V}}_N\succ0)\to1$, which is the event
on which parts (c) and (d) are stated. As previously pointed out, the sequence $\{R_N\}$
is arbitrary, and the bounds behind parts (a) and (b) do not involve it, so both
convergences hold uniformly over such sequences. Uniformity extends to the size statement
in part (d), over classes of data generating processes sharing the constants
$C_x$, $C_u$, $c_Q$ and $K$ of
Assumptions \ref{as:bound}, \ref{as:ident} and \ref{as:config} and a common vanishing
bound on $\delta_N$ (see the end of the proof of Theorem 1 in
\ref{sec:sma}).

Theorem \ref{thm:main} offers two advantages over available results. First, it dispenses
with independent clusters. Theorem 9 in \citet[p.~277]{hansen2019asymptotic} is also
rate-free, but it requires the observations to be partitioned into mutually independent
clusters, and a dyadic configuration need not admit a
growing number of them. Second, it tests joint restrictions
whose directions converge at different orders. The dyadic-robust $t$-statistic of
Theorem 3.1 in \citet[p.~676]{tabordmeehan2019inference} studentizes a single
coefficient under one common
convergence rate. Our proof of Theorem \ref{thm:main} uses the central limit theorem for dependency graphs of Theorem 2 in \citet[p.~307]{janson1988normal}. The following condition on dyad
multiplicity makes the jackknife estimator in
\eqref{eq:vjack} attain the same limit as the dyadic-robust one.

\begin{assumption}[Dyad multiplicity]\label{as:mult}
$\mathcal{J}_N\big/\lambda_{\min}(\Omega_N)\longrightarrow 0$.
\end{assumption}

Assumption \ref{as:mult} is introduced here for the jackknife alone. Since $\mathcal{J}_N\ge N$ by
\eqref{eq:degrees}, it imposes that $\Omega_N$ accumulate strictly faster than under
independent sampling, and it is the more demanding the more observations a single dyad
carries.

\begin{theorem}\label{thm:jack}
Let Assumptions \ref{as:dep}--\ref{as:mult} hold, and let $\{R_N\}$ be as in Theorem
\ref{thm:main}. With $\mathcal{V}_N:=R_N'V_NR_N$ and $\widehat{\mathcal{V}}^{J}_N := R_N'\widehat{V}^{J}R_N$,
as $G\to\infty$: (a) $\mathcal{V}_N^{-1/2}\widehat{\mathcal{V}}^{J}_N\mathcal{V}_N^{-1/2}\longrightarrow_p I_q$,
uniformly over restriction sequences as in Theorem \ref{thm:main}(b); (b)
$\widehat{\mathcal{V}}^{J}_N\succeq0$ on $\mathcal{E}^{J}_N$, and
$\mathbb{P}(\mathcal{E}^{J}_N)\to1$, $\mathbb{P}(\widehat{\mathcal{V}}^{J}_N\succ0)\to1$; (c) on
that event, $(\widehat{\mathcal{V}}^{J}_N)^{-1/2}R_N'(\widehat{\boldsymbol{\beta}}-\boldsymbol{\beta})
\longrightarrow_d N(\mathbf{0},I_q)$; (d) under $H_0: R_N'\boldsymbol{\beta}=\mathbf{r}$,
\[
W^{J}_N := \big(R_N'\widehat{\boldsymbol{\beta}}-\mathbf{r}\big)'
\big(\widehat{\mathcal{V}}^{J}_N\big)^{-1}\big(R_N'\widehat{\boldsymbol{\beta}}-\mathbf{r}\big)
\;\longrightarrow_d\; \chi^2_q .
\]
\end{theorem}

It is important to note that part (b) of Theorem \ref{thm:jack} is a finite-sample property rather than an asymptotic statement. Positive semidefiniteness holds at every $N$ by construction, following directly from the outer-product form of \eqref{eq:vjack}. In contrast, Theorem \ref{thm:main}(c) lacks this finite-sample guarantee, since 
$\widehat{\Omega}_N$ is a difference of two positive semidefinite matrices (see Lemma \ref{lem:leaveout}(d) in the supplement), and Assumptions \ref{as:dep}--\ref{as:config} leave this difference unrestricted. Therefore it is instructive to precisely delineate the role of Assumption \ref{as:mult}. The positive semidefiniteness in part (b) does not rely on it, holding unconditionally at every $N$. What Assumption \ref{as:mult} actually delivers is the asymptotic convergence established in parts (a) and (d). Without this condition, the jackknife estimator need not symmetrically attain 
$\mathcal{V}_N$, leaving in its place only the one-sided inequality described by Proposition \ref{prop:onesided} below. Consequently, the following proposition formally establishes that the error of the jackknife estimator is strictly one-sided.

\begin{proposition}\label{prop:onesided}
Let Assumptions \ref{as:dep}--\ref{as:config} hold and let $\{R_N\}$ be as in Theorem
\ref{thm:main}. Then there is a sequence $\varepsilon_N=o_p(1)$, not depending on $R_N$,
such that
$\lambda_{\min}\big(\mathcal{V}_N^{-1/2}\widehat{\mathcal{V}}^{J}_N\mathcal{V}_N^{-1/2}\big)\ge1-\varepsilon_N$.
Consequently, for every $\gamma\in(0,1)$,
$\limsup_{G\to\infty}\mathbb{P}\big(W^{J}_N>\chi^2_{q,1-\gamma}\big)\le\gamma$ under
$H_0:R_N'\boldsymbol{\beta}=\mathbf{r}$.
\end{proposition}

That a jackknife variance estimator errs upward is classical, going back to \citet[p.~590]{efron1981jackknife}, and it has been established for network functionals by \citet[p.~6107]{lin2020theoretical}. Note that the only difference between Proposition \ref{prop:onesided} and those results is that the error is signed in the positive semidefinite order, and that the bound holds uniformly over the restrictions tested. Proposition \ref{prop:onesided} does not impose Assumption \ref{as:mult}, which is what turns its inequality into the equality of Theorem \ref{thm:jack}(a). In particular, if Assumption \ref{as:mult} fails, a test that rejects when $W^{J}_N>\chi^2_{q,1-\gamma}$ still has asymptotic size no larger than $\gamma$.

Finally, note that Assumptions \ref{as:dep}--\ref{as:config} make $\mathbb{P}(\mathcal{E}_N)\to1$, but not $\mathcal{E}_N$ an almost-sure event at any fixed $N$. We therefore set $W_N:=0$ and $W^{J}_N:=0$ whenever the matrix required to compute each statistic is singular. Both statistics are then
defined on every sample, and the asymptotic size claims of Theorem \ref{thm:main}(d), Theorem \ref{thm:jack}(d) and Proposition \ref{prop:onesided} are ordinary probability statements; see, e.g., Remark 3.4 in \citet[p.~676]{tabordmeehan2019inference} for other alternatives.

\section{Discussion and Concluding Remarks}
\label{sec:conclusion}

In this paper we have considered Wald inference for linear regression on dyadic data in which several observations may share the same pair of units. We have shown that the test built from the dyadic-robust variance estimator is asymptotically valid for an arbitrary sequence of full-rank nonrandom joint restrictions, under a single condition on the accumulation of
dependence and with no convergence rate assumed or estimated. We have also shown that its delete-one-unit jackknife counterpart, which is positive semidefinite at every sample size, attains the same $\chi^2_q$ limiting distribution under one further condition on dyad multiplicity, and that it never over-rejects asymptotically when that condition fails. The
practical content is that the estimator applied researchers already compute
\citep{fafchamps2007formation,aronow2015cluster,graham2020network} supports joint hypothesis testing exactly as computed, on directed, multilayer, and panel dyadic designs alike, without the common-rate normalization that its existing formal justification requires. Appendix \ref{sec:mc2} puts that last point on a design in which the coefficients tested jointly do converge at three different rates, and finds the size of the test governed by the number of units rather than by the spread of those rates. The analogous rate-free result for clustered data is Theorem 9 of \citet[p.~277]{hansen2019asymptotic}, which relies fundamentally on mutually independent clusters. Since dyadic data inherently violate this independence, our Theorems \ref{thm:main} and \ref{thm:jack} substitute it with a dyadic dependency graph. The two sets of conditions are not nested: while their sampling scheme is a special case of our dyadic configuration, Assumptions \ref{as:bound} and \ref{as:acc} neither imply nor are implied by their cluster-size and moment regimes.

Three extensions are left for future research. The first is nonlinear estimation. The
Poisson pseudo-maximum-likelihood estimator of \citet{santossilva2006log} has a score that
is again a sum of dyad-indexed terms, so the route to Theorem \ref{thm:main} appears open.
The leave-unit-out algebra behind Theorem \ref{thm:jack} would there be replaced by an
asymptotic expansion. The second is
resampling. A bootstrap counterpart to the tests developed here would have to reproduce
dependence accumulating at heterogeneous and unknown rates, and its first-order
validity would require a separate argument. The third is growing
dimension. Node fixed effects, or layer-specific slopes with a growing number of layers,
violate the fixed-$K$ part of Assumption \ref{as:config}, and every bound in the
supplement would need to be made uniform in $K$. Finally, our analysis conditions on the observed network configuration throughout. Modeling the underlying network formation process represents a fundamentally distinct problem that remains beyond the scope of this paper.

\spacingset{1}
\putbib[references]
\end{bibunit}

\clearpage
\setcounter{page}{1}
\setcounter{section}{0}
\setcounter{equation}{0}
\setcounter{table}{0}
\setcounter{remark}{0}
\appendix
\renewcommand\thesection{Appendix SM.\Alph{section}}
\renewcommand\thesubsection{SM.\Alph{section}.\arabic{subsection}}
\renewcommand{\theequation}{SM.\Alph{section}.\arabic{equation}}
\renewcommand{\thelemma}{SM.\Alph{section}.\arabic{lemma}}
\renewcommand{\thetheorem}{SM.\Alph{section}.\arabic{theorem}}
\renewcommand{\theremark}{SM.\Alph{section}.\arabic{remark}}
\renewcommand{\thetable}{SM.\Alph{section}.\arabic{table}}
\makeatletter
\@addtoreset{equation}{section}
\@addtoreset{lemma}{section}
\@addtoreset{theorem}{section}
\@addtoreset{remark}{section}
\@addtoreset{table}{section}
\makeatother

\begin{bibunit}[jpe]

\spacingset{1}
\begin{center}
{\bf\Large Rate-Agnostic Wald Inference for Dyadic Regressions}\\[6pt]
{\large --- Supplemental Materials ---}\\[6pt]
\if0\blind
{\large Benjamin O. Harrison}\\[2pt]
{\small Department of Economics, Emory University, Rich Building 306,\\
1602 Fishburne Dr., Atlanta, GA 30322-2240, USA.\\
\faEnvelopeO: \href{mailto:benjamin.harrison@emory.edu}{benjamin.harrison@emory.edu}}\\[6pt]
{\large David T. Jacho-Ch\'{a}vez} (corresponding author)\\[2pt]
{\small Department of Economics, Emory University, Rich Building 306,\\
1602 Fishburne Dr., Atlanta, GA 30322-2240, USA.\\
\faEnvelopeO: \href{mailto:djachocha@emory.edu}{djachocha@emory.edu}}\\[6pt]
\fi
\today
\end{center}
\vspace{1em}

\spacingset{1.45}

\noindent Lemmas, equations, sections and tables numbered with an ``SM'' prefix are
stated in this supplement; all other numbers refer to the main text. Throughout,
$C:=C_x^2C_u^2$ as in the main text, the estimators and the events $\mathcal{E}_N$,
$\mathcal{E}^{J}_N$ are those of Section \ref{sec:model} in the main text, and Assumptions
(A1)--(A6) are those of Section \ref{sec:assumptions}.

\section{Proofs of the Main Results}
\label{sec:sma}

In this section we provide mathematical proofs of all the main results in the main text.
The proofs of Theorems 1 and 2 and of Proposition 1 come first, and the lemmas they invoke
are collected after them, numbered in the order in which the main proofs first use
them.

\begin{proof}[Proof of Theorem 1]
Throughout put $\Xi_N:=(NQ_N)^{-1}\Omega_N^{1/2}$, so
$\Xi_N\Xi_N'=V_N$ by equation
\eqref{eq:vhat} in the main text; put
$\mathcal{R}_N:=(R_N'V_NR_N)^{-1/2}R_N'\Xi_N$, a $q\times K$ matrix with
$\mathcal{R}_N\mathcal{R}_N'=I_q$ and $\|\mathcal{R}_N\|=1$; and put
$Z_N:=\Omega_N^{-1/2}S_N$, which has $\|Z_N\|=O_p(1)$ since
$\mathbb{E}\|Z_N\|^2=K$.

(a) Since $\widehat{\boldsymbol{\beta}}-\boldsymbol{\beta}=(N\widehat{Q}_N)^{-1}S_N$ and
$(N\widehat{Q}_N)^{-1}=(NQ_N)^{-1}Q_N\widehat{Q}_N^{-1}$,
\begin{equation}\label{eq:decomp}
\mathcal{V}_N^{-1/2}R_N'\big(\widehat{\boldsymbol{\beta}}-\boldsymbol{\beta}\big)=
\mathcal{R}_NZ_N+\mathcal{R}_N\big(\Gamma_N-I_K\big)Z_N ,
\end{equation}
with $\Gamma_N$ the matrix of Lemma \ref{lem:bread}. For the leading term, fix a unit
vector $\mathbf{c}\in\mathbb{R}^{q}$; then
$\|\mathcal{R}_N'\mathbf{c}\|=1$, so $\{\mathcal{R}_N'\mathbf{c}\}$ is a sequence of unit
vectors and Lemma \ref{lem:clt} gives
$\mathbf{c}'\mathcal{R}_NZ_N\to_dN(0,1)$; the Cram\'{e}r--Wold device yields
$\mathcal{R}_NZ_N\to_dN(\mathbf{0},I_q)$. For the remainder,
$\|\mathcal{R}_N(\Gamma_N-I_K)Z_N\|\le1\cdot o_p(1)\cdot O_p(1)=o_p(1)$ by Lemma
\ref{lem:bread}, and Slutsky's theorem applied to equation \eqref{eq:decomp} gives (a).
Both bounds are free of $R_N$.

(b) Put $\widehat{\Lambda}_N:=\Omega_N^{-1/2}\widehat{\Omega}_N\Omega_N^{-1/2}$; Lemmas
\ref{lem:infeasible} and \ref{lem:residual} give $\widehat{\Lambda}_N\to_pI_K$. Since
$\Xi_N^{-1}=\Omega_N^{-1/2}NQ_N$,
$\Xi_N^{-1}\widehat{V}_N\Xi_N^{-1\prime}=\Gamma_N\widehat{\Lambda}_N\Gamma_N'$, and the matrix
$U_N:=\Xi_N^{-1}V_N^{1/2}$ is orthogonal, so by invariance of the operator norm under
orthogonal conjugation,
\begin{equation}\label{eq:vratio}
\big\|V_N^{-1/2}\widehat{V}_NV_N^{-1/2}-I_K\big\|=
\big\|\Gamma_N\widehat{\Lambda}_N\Gamma_N'-I_K\big\|\longrightarrow_p0 ,
\end{equation}
the convergence by writing $\Gamma\widehat{\Lambda}\Gamma'-I=
(\Gamma-I)\widehat{\Lambda}\Gamma'+(\widehat{\Lambda}-I)\Gamma'+(\Gamma'-I)$ and using
Lemma \ref{lem:bread}. Finally set $R_N^{*}:=V_N^{1/2}R_N(R_N'V_NR_N)^{-1/2}$, which has
$R_N^{*\prime}R_N^{*}=I_q$; then
$\mathcal{V}_N^{-1/2}\widehat{\mathcal{V}}_N\mathcal{V}_N^{-1/2}-I_q=
R_N^{*\prime}\big(V_N^{-1/2}\widehat{V}_NV_N^{-1/2}-I_K\big)R_N^{*}$, whose operator norm
is at most that in equation \eqref{eq:vratio}, uniformly in $R_N$.

(c) Let $\eta_N:=\|\mathcal{V}_N^{-1/2}\widehat{\mathcal{V}}_N\mathcal{V}_N^{-1/2}-I_q\|=o_p(1)$ by (b). On
$\{\eta_N\le1/2\}$, whose probability tends to one,
$\mathcal{V}_N^{-1/2}\widehat{\mathcal{V}}_N\mathcal{V}_N^{-1/2}$ has all eigenvalues in $[1/2,3/2]$, hence
$\widehat{\mathcal{V}}_N\succ0$. On that event Lemma \ref{lem:matrix} gives
$\|\widehat{\mathcal{V}}_N^{-1/2}\mathcal{V}_N^{1/2}-I_q\|_F\le\sqrt q\,\eta_N+2\sqrt{q\eta_N}\to_p0$, and
since $\widehat{\mathcal{V}}_N^{-1/2}R_N'(\widehat{\boldsymbol{\beta}}-\boldsymbol{\beta})=
(\widehat{\mathcal{V}}_N^{-1/2}\mathcal{V}_N^{1/2})\,\mathcal{V}_N^{-1/2}R_N'
(\widehat{\boldsymbol{\beta}}-\boldsymbol{\beta})$, part (a) and Slutsky's theorem give
(c).

(d) Under $H_0$, put
$\zeta_N:=\mathcal{V}_N^{-1/2}R_N'(\widehat{\boldsymbol{\beta}}-\boldsymbol{\beta})$ and
$\mathcal{C}_N:=\mathcal{V}_N^{-1/2}\widehat{\mathcal{V}}_N\mathcal{V}_N^{-1/2}$. Exactly, and with no square root of
$\widehat{\mathcal{V}}_N$ appearing, $W_N=\zeta_N'\mathcal{C}_N^{-1}\zeta_N$. By (b),
$\mathcal{C}_N^{-1}\to_pI_q$ on the event where $\mathcal{C}_N\succ0$; by (a),
$\|\zeta_N\|=O_p(1)$; hence
$|W_N-\zeta_N'\zeta_N|\le\|\mathcal{C}_N^{-1}-I_q\|\|\zeta_N\|^2=o_p(1)$ and
$\zeta_N'\zeta_N\to_d\chi^2_q$ gives (d).

It remains to record the uniform version of (d) claimed in Section \ref{sec:assumptions}.
Let $\mathcal{F}_N(\bar\delta_N)$ be the class of data generating processes satisfying
(A1)--(A3) and (A5) with the same constants $C_x$, $C_u$, $c_Q$ and $K$ and with
$\delta_N\le\bar\delta_N$, where $\bar\delta_N\to0$. Every bound used above is explicit in
those constants and in $\delta_N$ alone, so each step holds uniformly over
$\mathcal{F}_N(\bar\delta_N)$ as written. Suppose the conclusion failed. Then for some
$\varepsilon>0$ there would be a subsequence $N_k\uparrow\infty$ and
$P_k\in\mathcal{F}_{N_k}(\bar\delta_{N_k})$ with
$|\mathbb{P}_{P_k}(W_{N_k}>\chi^2_{q,1-\gamma})-\gamma|\ge\varepsilon$ for every $k$.
Reading $\{P_k\}$ as a single sequence of data generating processes along $N_k$, it
satisfies (A1)--(A3) and (A5) with those constants, while
$\delta_{N_k}\le\bar\delta_{N_k}\to0$ delivers (A4); part (d) therefore applies to it and
$W_{N_k}\to_d\chi^2_q$. Continuity of the $\chi^2_q$ distribution function at
$\chi^2_{q,1-\gamma}$ forces
$\mathbb{P}_{P_k}(W_{N_k}>\chi^2_{q,1-\gamma})\to\gamma$, a contradiction. The same
argument applies to $W^{J}_N$ and Theorem 2(d) over a class that in addition bounds
$\mathcal{J}_N/\lambda_{\min}(\Omega_N)$. No rate of convergence to $\gamma$ is claimed.
\hfill\end{proof}

\begin{proof}[Proof of Theorem 2]
Recall $\Xi_N:=(NQ_N)^{-1}\Omega_N^{1/2}$ and that $U_N:=\Xi_N^{-1}V_N^{1/2}$ is orthogonal,
so $\|V_N^{-1/2}BV_N^{-1/2}\|=\|\Xi_N^{-1}B\Xi_N^{-1\prime}\|$ for every symmetric $B$. Since
$\Xi_N^{-1}A^{-1}=\Omega_N^{-1/2}Q_N\widehat{Q}_N^{-1}=\Gamma_N\Omega_N^{-1/2}$ with
$\Gamma_N$ exactly the matrix of Lemma \ref{lem:bread}, for every symmetric $B$,
\begin{equation}\label{eq:metric}
\Big\|V_N^{-1/2}A^{-1}BA^{-1}V_N^{-1/2}\Big\|=
\Big\|\Gamma_N\Omega_N^{-1/2}B\,\Omega_N^{-1/2}\Gamma_N'\Big\|\;\le\;
\|\Gamma_N\|^2\,\frac{\|B\|}{\lambda_{\min}(\Omega_N)} .
\end{equation}
Applying equation \eqref{eq:metric} to each of Lemma \ref{lem:remainders}'s blocks and
substituting $\rho=4C_x^2\mathcal{M}^{H}/(c_QN)$: the $B_2$ and $B_3$ bounds become
$96CC_x^2(\mathcal{M}^{H})^{2}/c_Q$ and at most $256CC_x^4(\mathcal{M}^{H})^{2}/c_Q^2$,
each bounded after division by $\lambda_{\min}(\Omega_N)$ by a constant times
$(\mathcal{M}^{H})^{2}/\lambda_{\min}(\Omega_N)\le\delta_N^{1/2}$ by Lemma
\ref{lem:struct}(e)(ii); and for $B_4$, by equation \eqref{eq:degrees} in the main text,
$N\mathcal{M}^{H}/\lambda_{\min}(\Omega_N)=\delta_N^{1/2}(N/\mathcal{M}^{H})^{1/2}\le
\delta_N^{1/2}(G/2)^{1/2}$, so
$32CN\mathcal{M}^{H}/\{G\lambda_{\min}(\Omega_N)\}\le(32C/\sqrt2)\delta_N^{1/2}$.
Therefore, on the event $\mathcal{G}_N$ of Lemma \ref{lem:remainders} and for
every $N$ satisfying the degree condition
\eqref{eq:tencond} stated there,
\begin{equation}\label{eq:master}
\Big\|V_N^{-1/2}\big(\widehat{V}^{J}-\widehat{V}_N\big)V_N^{-1/2}\Big\|\;\le\;
\|\Gamma_N\|^2\left[\frac{4C\,\mathcal{J}_N}{\lambda_{\min}(\Omega_N)}+
C_{\ast}\,\delta_N^{1/2}\right],\qquad
C_{\ast}:=\frac{96CC_x^2}{c_Q}+\frac{256CC_x^4}{c_Q^2}+\frac{32C}{\sqrt2},
\end{equation}
a constant depending only on $(C_x,C_u,c_Q)$. This display is the whole content of the
theorem: the bracket's second term vanishes under (A4) alone, and its first term is
exactly what (A6) is for. By Lemma \ref{lem:bread}, $\|\Gamma_N\|^2\le4$ on an event of
probability tending to one, so combining equation \eqref{eq:master} with (A4) and (A6),
\begin{equation}\label{eq:vjconv}
\Big\|V_N^{-1/2}\big(\widehat{V}^{J}-\widehat{V}_N\big)V_N^{-1/2}\Big\|
\longrightarrow_p0 .
\end{equation}

(a) With $R_N^{*}:=V_N^{1/2}R_N\mathcal{V}_N^{-1/2}$, which has $R_N^{*\prime}R_N^{*}=I_q$,
$\mathcal{V}_N^{-1/2}\widehat{\mathcal{V}}^{J}_N\mathcal{V}_N^{-1/2}-\mathcal{V}_N^{-1/2}\widehat{\mathcal{V}}_N\mathcal{V}_N^{-1/2}=
R_N^{*\prime}\big(V_N^{-1/2}(\widehat{V}^{J}-\widehat{V}_N)V_N^{-1/2}\big)R_N^{*}$, whose
operator norm is at most that in equation \eqref{eq:vjconv}, uniformly in $R_N$; adding
Theorem 1(b) gives (a).

(b) $\mathbb{P}(\mathcal{E}^{J}_N)\to1$ is Lemma \ref{lem:leaveout}(a); on
$\mathcal{E}^{J}_N$, equation \eqref{eq:vjack} in the main text exhibits
$\widehat{V}^{J}$ as a nonnegative multiple of a sum of outer products, so
$\widehat{\mathcal{V}}^{J}_N\succeq0$ for every $R_N$. For strict definiteness, let
$\eta^{J}_N:=\|\mathcal{V}_N^{-1/2}\widehat{\mathcal{V}}^{J}_N\mathcal{V}_N^{-1/2}-I_q\|=o_p(1)$ by (a); on
$\{\eta^{J}_N\le1/2\}$ all eigenvalues lie in $[1/2,3/2]$, hence
$\widehat{\mathcal{V}}^{J}_N\succ0$.

(c) On the event of (b), Lemma \ref{lem:matrix} gives
$\|(\widehat{\mathcal{V}}^{J}_N)^{-1/2}\mathcal{V}_N^{1/2}-I_q\|_F\to_p0$, and Theorem 1(a), which
concerns $\widehat{\boldsymbol{\beta}}$ only and is untouched by the choice of variance
estimator, with Slutsky's theorem gives (c).

(d) Exactly as in the proof of Theorem 1(d), with
$\mathcal{C}^{J}_N:=\mathcal{V}_N^{-1/2}\widehat{\mathcal{V}}^{J}_N\mathcal{V}_N^{-1/2}$,
$W^{J}_N=\zeta_N'(\mathcal{C}^{J}_N)^{-1}\zeta_N$; by (a),
$(\mathcal{C}^{J}_N)^{-1}\to_pI_q$ on the event where $\mathcal{C}^{J}_N\succ0$, and
$\zeta_N'\zeta_N\to_d\chi^2_q$ gives (d).
\hfill\end{proof}

\begin{proof}[Proof of Proposition 1]
Assumption (A6) is not used anywhere in this proof. By Lemma \ref{lem:remainders}, on
its event $\mathcal{G}_N$ and for $N$ satisfying the
degree condition \eqref{eq:tencond},
\begin{align*}
V_N^{-1/2}\widehat{V}^{J}V_N^{-1/2}&=V_N^{-1/2}\widehat{V}_NV_N^{-1/2}+
V_N^{-1/2}A^{-1}B_1A^{-1}V_N^{-1/2}+E_N,\\
E_N:&=V_N^{-1/2}A^{-1}\big(B_2-B_3-B_4\big)A^{-1}V_N^{-1/2}.
\end{align*}
The bounds on $B_2,B_3,B_4$, fed through equation \eqref{eq:metric} and simplified as in
the proof of Theorem 2, give $\|E_N\|\le\|\Gamma_N\|^2C_{\ast}\delta_N^{1/2}\to_p0$ by
(A4) alone. The middle term is positive semidefinite because $B_1\succeq0$ and congruence
preserves positive semidefiniteness, and the first term tends to $I_K$ in probability by
equation \eqref{eq:vratio}. Writing
$\varepsilon_N:=\|V_N^{-1/2}\widehat{V}_NV_N^{-1/2}-I_K\|+\|E_N\|=o_p(1)$, it follows
that $V_N^{-1/2}\widehat{V}^{J}V_N^{-1/2}\succeq(1-\varepsilon_N)I_K$; congruence by
$R_N^{*}$ preserves the ordering and yields
$\mathcal{C}^{J}_N\succeq(1-\varepsilon_N)I_q$ uniformly in $R_N$, the first display of
the Proposition. For the second, on $\{\varepsilon_N<1\}$ the matrix
$\mathcal{C}^{J}_N$ is invertible with
$(\mathcal{C}^{J}_N)^{-1}\preceq(1-\varepsilon_N)^{-1}I_q$, so
$W^{J}_N\le\|\zeta_N\|^2/(1-\varepsilon_N)$. Let
$\mathcal{H}_N:=\mathcal{G}_N\cap\{\varepsilon_N<1\}$, so
$\mathbb{P}(\mathcal{H}_N)\to1$. Fix $\gamma\in(0,1)$; since
$\|\zeta_N\|^2\to_d\chi^2_q$ by Theorem 1(a) and $\varepsilon_N=o_p(1)$, Slutsky's
theorem gives $\|\zeta_N\|^2/(1-\varepsilon_N)\to_d\chi^2_q$, and splitting on
$\mathcal{H}_N$,
\[
\mathbb{P}\big(W^{J}_N>\chi^2_{q,1-\gamma}\big)\;\le\;
\mathbb{P}\!\left(\frac{\|\zeta_N\|^2}{1-\varepsilon_N}>\chi^2_{q,1-\gamma}\right)+
\mathbb{P}\big(\mathcal{H}_N^{c}\big)\;\longrightarrow\;\gamma+0 .
\]
On $\mathcal{H}_N^{c}$ nothing is assumed about $W^{J}_N$: it is a well-defined random
variable there by the convention of Section \ref{sec:assumptions} in the main text,
setting a Wald statistic to zero when the matrix it inverts is singular, and the event is
discarded whole.
\hfill\end{proof}

The following lemmas appear in the order the proofs above first invoke them.

\begin{lemma}\label{lem:bread}
Under Assumptions (A1)--(A5), on the event of Lemma \ref{lem:design} where
$\widehat{Q}_N$ is invertible, define
$\Gamma_N:=\Omega_N^{-1/2}Q_N\widehat{Q}_N^{-1}\Omega_N^{1/2}$. Then
$\|\Gamma_N-I_K\|\longrightarrow_p0$.
\end{lemma}

\begin{proof}
From $\Gamma_N-I_K=\Omega_N^{-1/2}(Q_N-\widehat{Q}_N)\widehat{Q}_N^{-1}\Omega_N^{1/2}$,
\[
\|\Gamma_N-I_K\|\;\le\;
\left(\frac{\lambda_{\max}(\Omega_N)}{\lambda_{\min}(\Omega_N)}\right)^{1/2}
\big\|\widehat{Q}_N-Q_N\big\|\,\big\|\widehat{Q}_N^{-1}\big\| .
\]
By Lemma \ref{lem:struct}(d) the condition number is at most
$2CN\mathcal{M}^{H}/\lambda_{\min}(\Omega_N)$; by Lemma \ref{lem:design},
$\|\widehat{Q}_N-Q_N\|=O_p((\mathcal{M}^{H}/N)^{1/2})$ and
$\|\widehat{Q}_N^{-1}\|\le2/c_Q$ with probability tending to one. The factors of $N$
cancel, leaving
$\|\Gamma_N-I_K\|=O_p\big(\{(\mathcal{M}^{H})^{2}/\lambda_{\min}(\Omega_N)\}^{1/2}\big)
\to_p0$ by Lemma \ref{lem:struct}(e)(ii).
\hfill\end{proof}

\begin{lemma}\label{lem:clt}
Under Assumptions (A1)--(A5), for every sequence of unit vectors
$\mathbf{b}_N\in\mathbb{R}^{K}$,
$\mathbf{b}_N'\Omega_N^{-1/2}S_N\longrightarrow_d N(0,1)$. In particular
$Z_N:=\Omega_N^{-1/2}S_N\to_d N(\mathbf{0},I_K)$.
\end{lemma}

\begin{proof}
Fix $\{\mathbf{b}_N\}$, put $\mathbf{a}_N:=\Omega_N^{-1/2}\mathbf{b}_N$ and
$X_{Nn}:=(\mathbf{a}_N'\mathbf{x}_n)u_n$, so
$\sum_nX_{Nn}=\mathbf{b}_N'\Omega_N^{-1/2}S_N$. By (A3),
$\mathbb{E}[X_{Nn}]=0$, and
$\sigma_N^2=\mathbf{b}_N'\Omega_N^{-1/2}\Omega_N\Omega_N^{-1/2}\mathbf{b}_N=1$ exactly.
Since $\|\mathbf{a}_N\|\le\lambda_{\min}(\Omega_N)^{-1/2}$, (A2) gives $|X_{Nn}|\le
C_xC_u\lambda_{\min}(\Omega_N)^{-1/2}=:\Phi_N$ almost surely. Let
$\mathcal{D}_N$ have vertex set
$\{1,\dots,N\}$ and an edge between $n\ne m$ whenever $\mathbf{1}_{nm}=1$; if
$\mathcal{U}_1,\mathcal{U}_2$ are disjoint with no edge between them, then (A1) makes the
families $\{(\mathbf{x}_n,u_n)\}_{\mathcal{U}_1}$ and
$\{(\mathbf{x}_m,u_m)\}_{\mathcal{U}_2}$ independent, so $\mathcal{D}_N$ is a
dependency graph
for $\{X_{Nn}\}$.

Theorem 2 in \citet[p.~307]{janson1988normal}, taken at its parameter value $4$,
states that if a family $\{X_{Nn}\}_{n\le N}$ of random variables admits a dependency graph
whose maximal degree is at most $D_N$, if $|X_{Nn}|\le\Phi_N$ almost surely for every $n$,
and if $\sigma_N^{2}:=\operatorname{Var}(\sum_nX_{Nn})$ satisfies
\begin{equation}\label{eq:janson}
L_N\;:=\;\frac{(N/D_N)^{1/4}D_N\Phi_N}{\sigma_N}\;\longrightarrow\;0,
\end{equation}
then $\sigma_N^{-1}\sum_n\{X_{Nn}-\mathbb{E}X_{Nn}\}\to_dN(0,1)$. It remains to verify
\eqref{eq:janson}.

By the proof of Lemma \ref{lem:struct}(b), excluding $m=n$, the ordinary maximal
degree $d_N$ of $\mathcal{D}_N$ satisfies $d_N\le2\mathcal{M}^{H}-2$; the theorem is applied with
$D_N:=\max\{1,d_N\}$, and since $\mathcal{M}^{H}\ge1$,
\begin{equation}\label{eq:degbound}
D_N=\max\{1,d_N\}\;\le\;2\mathcal{M}^{H}.
\end{equation}
The distinction matters: a configuration of pairwise disjoint dyads has no edges, so
$D_N=1$ while $2\mathcal{M}^{H}-2=0$; equation \eqref{eq:degbound} holds in every case
and is what is used. Since $(N/D_N)^{1/4}D_N=N^{1/4}D_N^{3/4}$ is increasing
in $D_N$, substituting equation \eqref{eq:degbound},
\begin{align*}
L_N&=\frac{(N/D_N)^{1/4}D_N\Phi_N}{\sigma_N}\le
C_xC_u\frac{N^{1/4}(2\mathcal{M}^{H})^{3/4}}{\lambda_{\min}(\Omega_N)^{1/2}},\\
\text{so}\quad L_N^2&\le2^{3/2}C\,\frac{N^{1/2}(\mathcal{M}^{H})^{3/2}}{\lambda_{\min}(\Omega_N)}
=2^{3/2}C\,\delta_N^{1/2}\longrightarrow0
\end{align*}
by (A4). That theorem gives the conclusion, and the final claim
follows by the Cram\'{e}r--Wold device with constant sequences.
\hfill\end{proof}

\begin{lemma}\label{lem:infeasible}
Under Assumptions (A1)--(A5),
$\|\Omega_N^{-1/2}(\widetilde{\Omega}_N-\Omega_N)\Omega_N^{-1/2}\|_F\longrightarrow_p0$.
\end{lemma}

\begin{proof}
For symmetric positive definite $\Omega$ and any $B$,
$\|\Omega^{-1/2}B\Omega^{-1/2}\|_F\le\|B\|_F/\lambda_{\min}(\Omega)$. Hence, by Lemma
\ref{lem:struct}(a),
\[
\mathbb{E}\big\|\Omega_N^{-1/2}\big(\widetilde{\Omega}_N-\Omega_N\big)
\Omega_N^{-1/2}\big\|_F^2\;\le\;
\frac{\sum_{a,b=1}^{K}\operatorname{Var}\big[(\widetilde{\Omega}_N)_{ab}\big]}
{\lambda_{\min}(\Omega_N)^2}.
\]
Fix $a,b$ and write $\xi_{ij}:=u_iu_jx_{ia}x_{jb}$, so
$\operatorname{Var}[(\widetilde{\Omega}_N)_{ab}]=\sum_{i,j,k,l}\mathbf{1}_{ij}
\mathbf{1}_{kl}\operatorname{Cov}(\xi_{ij},\xi_{kl})$. If
$(\psi_i\cup\psi_j)\cap(\psi_k\cup\psi_l)=\emptyset$, then (A1) applied with
$S_1=\{i,j\}$ and $S_2=\{k,l\}$ makes $\xi_{ij}$ and $\xi_{kl}$ independent, so only
$(i,j,k,l)\in S$ contribute; by (A2), $|\xi_{ij}|\le C$, so
$|\operatorname{Cov}(\xi_{ij},\xi_{kl})|\le2C^2$. With Lemma \ref{lem:struct}(c),
$\operatorname{Var}[(\widetilde{\Omega}_N)_{ab}]\le2C^2|S|\le36C^2N(\mathcal{M}^{H})^{3}$,
and summing over the $K^2$ entries the display above is bounded by
$36K^2C^2\delta_N\to0$ by (A4). Chebyshev's inequality completes the proof.
\hfill\end{proof}

\begin{lemma}\label{lem:residual}
Under Assumptions (A1)--(A5),
$\|\Omega_N^{-1/2}(\widehat{\Omega}_N-\widetilde{\Omega}_N)\Omega_N^{-1/2}\|_F
\longrightarrow_p0$.
\end{lemma}

\begin{proof}
Write $\boldsymbol{\tau}:=\widehat{\boldsymbol{\beta}}-\boldsymbol{\beta}$, so
$\widehat{u}_n\widehat{u}_m=u_nu_m-u_n(\boldsymbol{\tau}'\mathbf{x}_m)-
(\boldsymbol{\tau}'\mathbf{x}_n)u_m+(\boldsymbol{\tau}'\mathbf{x}_n)(\boldsymbol{\tau}'\mathbf{x}_m)$ and
$\widehat{\Omega}_N-\widetilde{\Omega}_N=R_2+R_3+R_4$ with
$R_2:=-\sum_{n,m}\mathbf{1}_{nm}u_n(\boldsymbol{\tau}'\mathbf{x}_m)\mathbf{x}_n\mathbf{x}_m'$,
$R_3=R_2'$, and
$R_4:=\sum_{n,m}\mathbf{1}_{nm}(\boldsymbol{\tau}'\mathbf{x}_n)(\boldsymbol{\tau}'\mathbf{x}_m)
\mathbf{x}_n\mathbf{x}_m'$. By Lemma \ref{lem:design}, on an event of probability tending
to one $\lambda_{\min}(\widehat{Q}_N)\ge c_Q/2$, whence
$\|\boldsymbol{\tau}\|\le2\|S_N\|/(c_QN)$; and since
$\mathbb{E}\|S_N\|^2=\operatorname{tr}(\Omega_N)\le K\lambda_{\max}(\Omega_N)\le
2KCN\mathcal{M}^{H}$ by Lemma \ref{lem:struct}(d), Markov's inequality gives
$\|S_N\|=O_p\big((N\mathcal{M}^{H})^{1/2}\big)$ and therefore
\begin{equation}\label{eq:dhat}
\|\boldsymbol{\tau}\|=O_p\Big(\big(\mathcal{M}^{H}/N\big)^{1/2}\Big),
\end{equation}
using only Lemma \ref{lem:design} and Markov's inequality, not Lemma \ref{lem:clt},
so there is no circularity. Since
$\|\mathbf{x}_n\mathbf{x}_m'\|_F\le C_x^2$ and $|\boldsymbol{\tau}'\mathbf{x}_m|\le
\|\boldsymbol{\tau}\|C_x$, (A2) and Lemma \ref{lem:struct}(b) give
$\|R_2\|_F\le2C_uC_x^3N\mathcal{M}^{H}\|\boldsymbol{\tau}\|$ and
$\|R_4\|_F\le2C_x^4N\mathcal{M}^{H}\|\boldsymbol{\tau}\|^2$. Dividing by
$\lambda_{\min}(\Omega_N)$ and substituting equation \eqref{eq:dhat},
$\|R_2\|_F/\lambda_{\min}(\Omega_N)=
O_p\big(N^{1/2}(\mathcal{M}^{H})^{3/2}/\lambda_{\min}(\Omega_N)\big)=
O_p(\delta_N^{1/2})$ and
$\|R_4\|_F/\lambda_{\min}(\Omega_N)=
O_p\big((\mathcal{M}^{H})^{2}/\lambda_{\min}(\Omega_N)\big)$; the first tends to zero by
(A4) and the second by Lemma \ref{lem:struct}(e)(ii). Applying
$\|\Omega^{-1/2}B\Omega^{-1/2}\|_F\le\|B\|_F/\lambda_{\min}(\Omega)$ as in Lemma
\ref{lem:infeasible} completes the proof.
\hfill\end{proof}

\begin{lemma}\label{lem:matrix}
Let $M,\widehat{M}$ be symmetric positive definite $q\times q$ matrices and put
$\mathcal{C}:=M^{-1/2}\widehat{M}M^{-1/2}$ and $\eta:=\|\mathcal{C}-I_q\|$. If
$\eta\le1/2$, then $\Psi:=\widehat{M}^{-1/2}M^{1/2}$ satisfies
$\|\Psi-I_q\|_F\le\sqrt q\,\eta+2\sqrt{q\eta}$.
\end{lemma}

\begin{proof}
$\Psi'\Psi=M^{1/2}\widehat{M}^{-1}M^{1/2}=\mathcal{C}^{-1}$, whose eigenvalues lie in
$[(1+\eta)^{-1},(1-\eta)^{-1}]$, so every singular value of $\Psi$ lies in
$[(1+\eta)^{-1/2},(1-\eta)^{-1/2}]$; for $\eta\le1/2$ one checks
$1-(1+\eta)^{-1/2}\le\eta$ and $(1-\eta)^{-1/2}-1\le\eta$, so
$\varepsilon:=\max_i|\sigma_i(\Psi)-1|\le\eta$. Next, $\Psi$ is the product of
the symmetric
positive definite matrices $\widehat{M}^{-1/2}$ and $M^{1/2}$; conjugating by
$\widehat{M}^{1/4}$ gives
$\widehat{M}^{1/4}\Psi\widehat{M}^{-1/4}=\widehat{M}^{-1/4}M^{1/2}\widehat{M}^{-1/4}\succ0$,
so $\Psi$ is similar to a
symmetric positive definite matrix and its eigenvalues are real and positive, and since
$\sigma_{\min}(\Psi)\le|\lambda_i(\Psi)|\le\sigma_{\max}(\Psi)$ they lie in
$[1-\varepsilon,1+\varepsilon]$. Finally write
$\Psi=U(\Theta+\Upsilon)U^{*}$ in Schur form
with $\Theta$ diagonal carrying the eigenvalues and $\Upsilon$ strictly upper
triangular;
$\Theta$ and $\Upsilon$ are orthogonal in the Frobenius inner product, so
$\|\Upsilon\|_F^2=\|\Psi\|_F^2-\|\Theta\|_F^2=
\sum_i\sigma_i(\Psi)^2-\sum_i\lambda_i(\Psi)^2\le
q[(1+\varepsilon)^2-(1-\varepsilon)^2]=4q\varepsilon$, whence
$\|\Psi-I_q\|_F\le\|\Theta-I_q\|_F+\|\Upsilon\|_F\le
\sqrt q\,\varepsilon+2\sqrt{q\varepsilon}
\le\sqrt q\,\eta+2\sqrt{q\eta}$.
\hfill\end{proof}

Note that rotation invariance of the limit law does not by itself permit $\Psi$ to be
replaced by an orthogonal matrix. A matrix whose Gram matrices both tend to $I_q$ may
itself tend to a nontrivial orthogonal matrix, and a rotation that is random and
correlated with the statistic destroys the limit. Taking $O_N$ to carry $Z_N$ to
$(\|Z_N\|,0,\dots,0)'$ concentrates $O_NZ_N$ on a ray. The positivity of the spectrum of
$\Psi$ established above is what rules this out.

For the remaining lemmas with regards to the jackknife, write $\mathbf{v}_n:=\mathbf{x}_n\widehat{u}_n$,
$T_g:=\sum_{n\in\mathcal{N}_g}\mathbf{v}_n$, and $D_d:=\sum_{n:\psi(n)=d}\mathbf{v}_n$,
so that $T_g$ aggregates the score residuals at unit $g$ and $D_d$ at dyad $d$; write
also $\rho_g:=4C_x^2M_g/(c_QN)$ and $\rho:=4C_x^2\mathcal{M}^{H}/(c_QN)$, and let
\begin{equation}\label{eq:tencond}
\mathcal{M}^{H}/N\;\le\;c_Q\big/\big(4C_x^{2}\big)
\end{equation}
denote the degree condition under which $\rho\le1$; by Lemma \ref{lem:struct}(e)(i),
equation \eqref{eq:tencond} holds for all $N$ large enough.

\begin{lemma}\label{lem:remainders}
Under Assumptions (A1)--(A5), on the event
$\mathcal{G}_N:=\{\lambda_{\min}(\widehat{Q}_N)\ge c_Q/2\}\cap
\{\|\widehat{\boldsymbol{\beta}}-\boldsymbol{\beta}\|\le C_u/C_x\}$, which has
$\mathbb{P}(\mathcal{G}_N)\to1$, and for every $N$ satisfying equation
\eqref{eq:tencond}, one has the exact decomposition
$\widehat{V}^{J}-\widehat{V}_N=A^{-1}(B_1+B_2-B_3-B_4)A^{-1}$ with $B_1,\dots,B_4$
symmetric, $B_1\succeq0$, $B_4\succeq0$, and
\[
\|B_1\|\le4C\,\mathcal{J}_N,\qquad
\|B_2\|\le24C\rho N\mathcal{M}^{H},\qquad
\|B_3\|\le\frac{16C\rho^2N^2}{G},\qquad
\|B_4\|\le\frac{32CN\mathcal{M}^{H}}{G}.
\]
\end{lemma}

\begin{proof}
On $\mathcal{G}_N$,
$|\widehat{u}_n|\le C_u+\|\widehat{\boldsymbol{\beta}}-\boldsymbol{\beta}\|C_x\le2C_u$ by
(A2), so $\|\mathbf{v}_n\|\le2C_uC_x$, $\|T_g\|\le2C_uC_xM_g$ and
$\|D_d\|\le2C_uC_xL_d$; by the handshake identity $\sum_gM_g=2N$ and
$M_g\le\mathcal{M}^{H}$,
\begin{equation}\label{eq:Tsums}
\sum_g\|T_g\|^2\le8CN\mathcal{M}^{H},\qquad\qquad
\sum_g\|T_g\|\le4C_uC_xN .
\end{equation}
That $\mathbb{P}(\mathcal{G}_N)\to1$ follows from Lemma \ref{lem:design} and from
equation \eqref{eq:dhat} together with Lemma \ref{lem:struct}(e)(i). Write $\mathbf{w}_g:=A_{(g)}^{-1}T_g$, so Lemma \ref{lem:leaveout}(b)
gives $\widehat{\boldsymbol{\beta}}_{(g)}-\bar{\boldsymbol{\beta}}=
-(\mathbf{w}_g-\boldsymbol{\omega})$ with $\boldsymbol{\omega}:=G^{-1}\sum_g\mathbf{w}_g$, and, with
$\Sigma:=\sum_g(\mathbf{w}_g-\boldsymbol{\omega})(\mathbf{w}_g-\boldsymbol{\omega})'$,
\begin{equation}\label{eq:sigma}
\widehat{V}^{J}=\frac{G-1}{G}\Sigma=\Sigma-\frac1G\Sigma,\qquad
\Sigma=\sum_g\mathbf{w}_g\mathbf{w}_g'-G\,\boldsymbol{\omega}\boldsymbol{\omega}' .
\end{equation}
By Lemma \ref{lem:leaveout}(a),
$A_{(g)}^{-1}=A^{-1}(I_K+A_gA_{(g)}^{-1})$, so with
$\mathbf{t}_g:=(I_K+A_gA_{(g)}^{-1})T_g$ we have $\mathbf{w}_g=A^{-1}\mathbf{t}_g$; and
by Lemma \ref{lem:leaveout}(c),
$\bar{\mathbf{t}}:=G^{-1}\sum_g\mathbf{t}_g=\boldsymbol{\pi}$ with
$\boldsymbol{\pi}:=G^{-1}\sum_gA_gA_{(g)}^{-1}T_g$ and
$\boldsymbol{\omega}=A^{-1}\boldsymbol{\pi}$, the
step that makes the centring second order, because each observation belongs to exactly
two deletion groups. Expanding
$\sum_g\mathbf{t}_g\mathbf{t}_g'=\sum_gT_gT_g'+B_2$ with
$B_2:=\sum_g[A_gA_{(g)}^{-1}T_gT_g'+T_gT_g'A_{(g)}^{-1}A_g+
A_gA_{(g)}^{-1}T_gT_g'A_{(g)}^{-1}A_g]$, symmetric; using Lemma \ref{lem:leaveout}(d)
and equation \eqref{eq:vhat} in the main text,
$A^{-1}(\sum_gT_gT_g')A^{-1}=\widehat{V}_N+A^{-1}B_1A^{-1}$ with
$B_1:=\sum_dD_dD_d'\succeq0$; and from equation \eqref{eq:sigma},
$B_3:=G\,\boldsymbol{\pi}\boldsymbol{\pi}'$ and
$B_4:=G^{-1}\sum_g(\mathbf{t}_g-\boldsymbol{\pi})(\mathbf{t}_g-\boldsymbol{\pi})'\succeq0$ collect
the remaining terms, giving the stated identity. For the bounds: $B_1\succeq0$ gives
$\|B_1\|\le\operatorname{tr}(B_1)=\sum_d\|D_d\|^2\le4C\sum_dL_d^2=4C\mathcal{J}_N$;
$\|B_2\|\le\sum_g(2\rho_g+\rho_g^2)\|T_g\|^2\le3\rho\sum_g\|T_g\|^2\le
24C\rho N\mathcal{M}^{H}$ by equation \eqref{eq:Tsums};
$\|\boldsymbol{\pi}\|\le\rho G^{-1}\sum_g\|T_g\|\le4\rho C_uC_xN/G$ gives
$\|B_3\|=G\|\boldsymbol{\pi}\|^2\le16C\rho^2N^2/G$; and
$\|B_4\|\le\operatorname{tr}(B_4)=G^{-1}\sum_g\|\mathbf{t}_g-\boldsymbol{\pi}\|^2\le
G^{-1}\sum_g\|\mathbf{t}_g\|^2\le4G^{-1}\sum_g\|T_g\|^2\le32CN\mathcal{M}^{H}/G$, using
$\|\mathbf{t}_g\|\le(1+\rho_g)\|T_g\|\le2\|T_g\|$ and that a sum of squared deviations
about the mean is no larger than the sum of squares.
\hfill\end{proof}

\begin{lemma}\label{lem:struct}
Under Assumptions (A1)--(A3) and (A5): (a)
$\mathbb{E}[\widetilde{\Omega}_N]=\Omega_N$ exactly. (b)
$|\mathcal{A}|\le N(2\mathcal{M}^{H}-1)\le2N\mathcal{M}^{H}$. (c) Let
$S:=\{(i,j,k,l):\mathbf{1}_{ij}=\mathbf{1}_{kl}=1,\
(\psi_i\cup\psi_j)\cap(\psi_k\cup\psi_l)\ne\emptyset\}$, where $\psi_i:=\psi(i)$; then
$|S|\le3N\mathcal{M}^{H}(2\mathcal{M}^{H}-1)(3\mathcal{M}^{H}-2)<18N(\mathcal{M}^{H})^{3}$.
(d) $\lambda_{\max}(\Omega_N)\le\|\Omega_N\|\le CN(2\mathcal{M}^{H}-1)\le
2CN\mathcal{M}^{H}$. (e) Under (A4) in addition, (i)
$\mathcal{M}^{H}/N\le4C^{2}\delta_N\to0$ and (ii)
$(\mathcal{M}^{H})^{2}/\lambda_{\min}(\Omega_N)\le\sqrt{\delta_N}\to0$.
\end{lemma}

\begin{proof}
(a) By (A3), $\operatorname{Cov}(\mathbf{x}_nu_n,\mathbf{x}_mu_m)=
\mathbb{E}[u_nu_m\mathbf{x}_n\mathbf{x}_m']$, so $\Omega_N=\sum_{n,m}
\mathbb{E}[u_nu_m\mathbf{x}_n\mathbf{x}_m']$. If $\mathbf{1}_{nm}=0$ then
$\psi(n)\cap\psi(m)=\emptyset$, and (A1) applied with $S_1=\{n\}$, $S_2=\{m\}$ makes
$(\mathbf{x}_n,u_n)$ independent of $(\mathbf{x}_m,u_m)$; that term then equals
$\mathbb{E}[\mathbf{x}_nu_n]\mathbb{E}[\mathbf{x}_mu_m]'=\mathbf{0}$. The indicator in
equation \eqref{eq:meat} in the main text therefore deletes only terms that vanish, and retains every term that need not.

(b) Fix $n$ with $\psi(n)=\{g,h\}$. Then $\{m:\mathbf{1}_{nm}=1\}=\{m:g\in\psi(m)\}\cup
\{m:h\in\psi(m)\}$, whose cardinality is $M_g+M_h-\#\{m:\psi(m)=\{g,h\}\}$ by
inclusion--exclusion, and the subtracted count is at least $1$ since $m=n$ qualifies.
Hence $\#\{m:\mathbf{1}_{nm}=1\}\le2\mathcal{M}^{H}-1$; summing over $n$ gives the claim.

(c) If $\mathbf{1}_{ij}=1$ then $|\psi_i\cup\psi_j|\le3$. For a unit $y$, put
$R_y:=\#\{(k,l):\mathbf{1}_{kl}=1,\ y\in\psi_k\cup\psi_l\}$, and write
$A_1:=\{(k,l):\mathbf{1}_{kl}=1,y\in\psi_k\}$, $A_2:=\{(k,l):\mathbf{1}_{kl}=1,
y\in\psi_l\}$, so $R_y=|A_1|+|A_2|-|A_1\cap A_2|$ with $|A_1|=|A_2|$ by symmetry. By part
(b)'s per-row bound, $|A_1|=\sum_{k:y\in\psi_k}\#\{l:\mathbf{1}_{kl}=1\}\le
M_y(2\mathcal{M}^{H}-1)$; and if $y\in\psi_k$ and $y\in\psi_l$ then $\mathbf{1}_{kl}=1$
automatically, so $|A_1\cap A_2|=M_y^2$ exactly. Hence $R_y\le
2M_y(2\mathcal{M}^{H}-1)-M_y^2$, which is increasing in $M_y$ on $M_y\le
2\mathcal{M}^{H}-1$, so $R_y\le\mathcal{M}^{H}(3\mathcal{M}^{H}-2)$. Every
$(i,j,k,l)\in S$ has some unit $y\in\psi_i\cup\psi_j$ lying in $\psi_k\cup\psi_l$;
summing over $(i,j)\in\mathcal{A}$ and the at most three units of $\psi_i\cup\psi_j$,
$|S|\le|\mathcal{A}|\cdot3\mathcal{M}^{H}(3\mathcal{M}^{H}-2)\le
3N(2\mathcal{M}^{H}-1)\mathcal{M}^{H}(3\mathcal{M}^{H}-2)<18N(\mathcal{M}^{H})^{3}$.

(d) By part (a)'s computation, $\Omega_N=\sum_{(n,m)\in\mathcal{A}}
\mathbb{E}[u_nu_m\mathbf{x}_n\mathbf{x}_m']$. Since $\mathbf{x}_n\mathbf{x}_m'$ has rank
one, $\|\mathbf{x}_n\mathbf{x}_m'\|=\|\mathbf{x}_n\|\|\mathbf{x}_m\|$, and by (A2) and
Jensen's inequality $\|\mathbb{E}[u_nu_m\mathbf{x}_n\mathbf{x}_m']\|\le C$. The triangle
inequality over $\mathcal{A}$ and part (b) give the claim; $K$ does not enter.

(e)(i) By (d), $\lambda_{\min}(\Omega_N)^2\le4C^2N^2(\mathcal{M}^{H})^{2}$, so
$\delta_N\ge N(\mathcal{M}^{H})^{3}/\{4C^2N^2(\mathcal{M}^{H})^{2}\}=
(4C^2)^{-1}\mathcal{M}^{H}/N$, a finite-sample inequality; (A4) gives the limit.
(e)(ii) By equation \eqref{eq:degrees} in the main text, $\mathcal{M}^{H}\le N$, so
$(\mathcal{M}^{H})^{4}\le N(\mathcal{M}^{H})^{3}$ and
$\{(\mathcal{M}^{H})^{2}/\lambda_{\min}(\Omega_N)\}^{2}\le\delta_N$.
\hfill\end{proof}

\begin{lemma}\label{lem:leaveout}
Under Assumptions (A1)--(A5): (a) if
$\lambda_{\min}(\widehat{Q}_N)\ge c_Q/2$ and $N$ satisfies equation \eqref{eq:tencond},
then $\lambda_{\min}(A_{(g)})\ge c_QN/4$ for every $g$, so $\mathcal{E}^{J}_N$ obtains;
both conditions hold with probability tending to one, so
$\mathbb{P}(\mathcal{E}^{J}_N)\to1$; and
$\|A_gA_{(g)}^{-1}\|\le\rho_g\le\rho\le1$, with
$A_{(g)}^{-1}-A^{-1}=A^{-1}A_gA_{(g)}^{-1}$. (b) On $\mathcal{E}^{J}_N$, for every $g$,
$\widehat{\boldsymbol{\beta}}_{(g)}-\widehat{\boldsymbol{\beta}}=-A_{(g)}^{-1}T_g$
exactly. (c) On $\mathcal{E}_N$, $\sum_{g=1}^{G}T_g=\mathbf{0}$. (d) Identically,
$\sum_{g=1}^{G}T_gT_g'=\widehat{\Omega}_N+\sum_dD_dD_d'$.
\end{lemma}

\begin{proof}
(a) By (A2), $\|A_g\|\le M_gC_x^2\le\mathcal{M}^{H}C_x^2$, while
$\lambda_{\min}(A)=N\lambda_{\min}(\widehat{Q}_N)\ge c_QN/2$; under equation
\eqref{eq:tencond}, $\mathcal{M}^{H}C_x^2\le c_QN/4$, so Weyl's inequality gives
$\lambda_{\min}(A_{(g)})\ge c_QN/4>0$ uniformly in $g$, whence
$\|A_{(g)}^{-1}\|\le4/(c_QN)$ and $\|A_gA_{(g)}^{-1}\|\le\rho_g\le\rho\le1$, the last
inequality being equation \eqref{eq:tencond} again; the resolvent identity
$A_{(g)}^{-1}-A^{-1}=A^{-1}(A-A_{(g)})A_{(g)}^{-1}=A^{-1}A_gA_{(g)}^{-1}$ is immediate.
The two conditions hold with
probability tending to one by Lemma \ref{lem:design} and by Lemma
\ref{lem:struct}(e)(i) respectively.

(b) Since $A=A_{(g)}+A_g$ with $A_g\succeq0$, invertibility of $A_{(g)}$ forces that of
$A$, so $\mathcal{E}^{J}_N\subseteq\mathcal{E}_N$; the normal equation then gives
$A_{(g)}\widehat{\boldsymbol{\beta}}_{(g)}=A\widehat{\boldsymbol{\beta}}-
\sum_{n\in\mathcal{N}_g}\mathbf{x}_ny_n$; substituting $A=A_{(g)}+A_g$,
$A_{(g)}(\widehat{\boldsymbol{\beta}}_{(g)}-\widehat{\boldsymbol{\beta}})=
-\sum_{n\in\mathcal{N}_g}\mathbf{x}_n(y_n-\widehat{\boldsymbol{\beta}}'\mathbf{x}_n)=-T_g$.
No expansion is involved and no remainder is discarded.

(c) Each observation lies in exactly two of the sets $\mathcal{N}_1,\dots,\mathcal{N}_G$,
so $\sum_gT_g=2\sum_n\mathbf{x}_n\widehat{u}_n=\mathbf{0}$ by the normal equations on
$\mathcal{E}_N$.

(d) Interchanging summation, $\sum_gT_gT_g'=\sum_{n,m}|\psi(n)\cap\psi(m)|\,
\mathbf{v}_n\mathbf{v}_m'$, and since $|\psi(n)|=|\psi(m)|=2$ the intersection has
cardinality $0$, $1$ or $2$, with $2$ exactly when $\psi(n)=\psi(m)$; hence
$|\psi(n)\cap\psi(m)|=\mathbf{1}_{nm}+\mathbf{1}\{\psi(n)=\psi(m)\}$ in all three cases,
and splitting the sum gives $\widehat{\Omega}_N$ from the first term, by equation
\eqref{eq:meat} in the main text, and $\sum_dD_dD_d'$ from the second.
\hfill\end{proof}

\begin{lemma}\label{lem:design}
Under Assumptions (A1)--(A5),
$\|\widehat{Q}_N-Q_N\|_F=O_p\big((\mathcal{M}^{H}/N)^{1/2}\big)=o_p(1)$ and
$\mathbb{P}\big(\lambda_{\min}(\widehat{Q}_N)\ge c_Q/2\big)\to1$.
\end{lemma}

\begin{proof}
Put $\Delta_n:=\mathbf{x}_n\mathbf{x}_n'-\mathbb{E}[\mathbf{x}_n\mathbf{x}_n']$, so
$\mathbb{E}\Delta_n=0$ and $\|\Delta_n\|_F\le2C_x^2$ by (A2). If $\mathbf{1}_{nm}=0$, (A1) makes
$\Delta_n$ and $\Delta_m$ independent and $\mathbb{E}\operatorname{tr}(\Delta_n\Delta_m')$ vanishes;
otherwise $|\mathbb{E}\operatorname{tr}(\Delta_n\Delta_m')|\le4C_x^4$. By Lemma
\ref{lem:struct}(b),
$\mathbb{E}\|\widehat{Q}_N-Q_N\|_F^2\le4C_x^4|\mathcal{A}|/N^2\le8C_x^4\mathcal{M}^{H}/N$,
which tends to $0$ by Lemma \ref{lem:struct}(e)(i); Markov's inequality gives the rate.
By Weyl's inequality,
$\lambda_{\min}(\widehat{Q}_N)\ge\lambda_{\min}(Q_N)-\|\widehat{Q}_N-Q_N\|\ge c_Q-o_p(1)$,
using (A3) and $\|\cdot\|\le\|\cdot\|_F$.
\hfill\end{proof}

\section{Three Illustrative Examples}
\label{sec:smb}

In this section we describe the framework on three designs that a practitioner might encounter in
applied work. Each illustration proceeds in three steps. We first write the design down. We then give a
symbol-by-symbol dictionary whose entries are the objects of Section \ref{sec:model} in the
main text, taken in the order that section introduces them, so that each entry says what
that piece of notation is in the application. We close with what Assumption (A1) leaves
free and what Assumption (A4) asks of the design. The numerical exercise of
\ref{sec:numerical} below instantiates the first two of them: the directed Monte Carlo
experiments and the empirical application are instances of Illustration 1, with $K_1=2$
and $K_1=4$ respectively, and the multilayer Monte Carlo experiments are an instance of
Illustration 2, with $L=3$ and $K_1=2$.

\subsection{Illustration 1: Directed Dyads}
\label{sec:smb-directed}

Let there be $G$ countries, indexed $g=1,\dots,G$, and let a row of the regression be an
\emph{ordered} pair $n=(g,h)$ with $g\ne h$, read as a flow from origin $g$ to
destination $h$, such as exports, migrants or capital. Every ordered pair is observed, so
\[
N=G(G-1), \qquad \psi\big((g,h)\big):=\{g,h\} .
\]
The orientation is carried by the row and forgotten by $\psi$. Exactly two rows, $(g,h)$
and $(h,g)$, carry each dyad, so $\psi$ is two-to-one, and this is the only feature of the
design the framework is told.

Let the outcome be $y_n=y_{gh}:=\ln T_{gh}$, the log flow from $g$ to $h$. Let
$\mathbf{z}_g\in\mathbb{R}^{K_1}$ collect country attributes, such as log GDP and log GDP
per capita, and let $\mathbf{s}_{\{g,h\}}\in\mathbb{R}^{K_2}$ collect symmetric pair
attributes, such as log distance, contiguity and the constant, with
$\mathbf{s}_{\{g,h\}}=\mathbf{s}_{\{h,g\}}$ by construction. The regressor stacks the
origin, destination and pair blocks,
\[
\mathbf{x}_n=\mathbf{x}_{gh}:=\big(\mathbf{z}_g',\ \mathbf{z}_h',\
\mathbf{s}_{\{g,h\}}'\big)'\in\mathbb{R}^{K}, \qquad K=2K_1+K_2 ,
\]
so that with $\boldsymbol{\beta}=(\boldsymbol{\beta}_O',\boldsymbol{\beta}_D',
\boldsymbol{\beta}_S')'$ the single equation $y_n=\boldsymbol{\beta}'\mathbf{x}_n+u_n$
is the gravity equation \eqref{eq:gravity} of \ref{sec:numerical} below, in which
$\boldsymbol{\beta}_O$ is the push effect of an attribute at the origin and
$\boldsymbol{\beta}_D$ its pull effect at the destination. Note that $\psi$ is symmetric
while $\mathbf{x}_n$ is not. Reversing the arrow swaps the first two blocks and fixes the
third,
\[
\mathbf{x}_{hg}=\Pi\,\mathbf{x}_{gh},\qquad
\Pi:=\begin{pmatrix}0&I_{K_1}&0\\ I_{K_1}&0&0\\ 0&0&I_{K_2}\end{pmatrix},
\qquad \Pi=\Pi',\quad \Pi^2=I_K .
\]
The hypothesis that push equals pull, $H_0:\boldsymbol{\beta}_O=\boldsymbol{\beta}_D$,
is
\[
R_N:=\begin{pmatrix}I_{K_1}\\ -I_{K_1}\\ 0_{K_2\times K_1}\end{pmatrix}
\in\mathbb{R}^{K\times q},\qquad q=K_1,\qquad \mathbf{r}=\mathbf{0} ;
\]
Because $\Pi$ is a symmetric involution it splits $\mathbb{R}^{K}$ orthogonally, and the columns of $R_N$ span its $(-1)$-eigenspace $\mathcal{P}_-$. The restriction
$R_N'\boldsymbol{\beta}=\mathbf{0}$ therefore says that $\boldsymbol{\beta}$ lies in the $(+1)$-eigenspace, $\Pi\boldsymbol{\beta}=\boldsymbol{\beta}$, and the two directions of the equivalence draw on different things. Under $H_0$,
$\boldsymbol{\beta}'\mathbf{x}_{hg}=\boldsymbol{\beta}'\Pi\mathbf{x}_{gh}=
(\Pi\boldsymbol{\beta})'\mathbf{x}_{gh}=\boldsymbol{\beta}'\mathbf{x}_{gh}$ for every ordered pair, whatever the attributes are. Conversely
$\boldsymbol{\beta}'(\mathbf{x}_{gh}-\mathbf{x}_{hg})=
(\boldsymbol{\beta}_O-\boldsymbol{\beta}_D)'(\mathbf{z}_g-\mathbf{z}_h)$, so the regression function forgets the orientation only if $\boldsymbol{\beta}_O-\boldsymbol{\beta}_D$ annihilates every attribute difference. Assumption (A3) already forbids that: $Q_N$ leaves $\mathcal{P}_-$ invariant, and at the unit vector $(\mathbf{v}',-\mathbf{v}',\mathbf{0}')'/\sqrt2$ of $\mathcal{P}_-$ indexed by $\mathbf{v}\in\mathbb{R}^{K_1}$ with $\|\mathbf{v}\|=1$, its quadratic form is $(2N)^{-1}\sum_{g\neq h}\mathbb{E}[\{\mathbf{v}'(\mathbf{z}_g-\mathbf{z}_h)\}^2]$, so $\lambda_{\min}(Q_N)\geq c_Q$ leaves no such direction. That is, the hypothesis is exactly that the regression function forgets the orientation, as $\psi$ does, and no condition beyond (A3) is needed to say so.

Each entry of the dictionary below is the corresponding object of Section
\ref{sec:model} in the main text.

\begin{description}
\item[{$g$:}] a country, $g=1,\dots,G$, and all limits are taken as $G\to\infty$.
\item[{$n$, $N$:}] an ordered pair $(g,h)$ with $g\ne h$, that is, one directed flow, and
  $N=G(G-1)$.
\item[{$\psi(n)$:}] the unordered pair $\{g,h\}$, so that the orientation is discarded and
  $\psi$ is two-to-one.
\item[{$M_g$, $\mathcal{M}^{H}$, $\mathcal{M}^{L}$:}] $M_g=2(G-1)$ for every $g$, since $g$
  is the origin in $G-1$ rows and the destination in $G-1$ others. Hence
  $\mathcal{M}^{H}=\mathcal{M}^{L}=2(G-1)$ and the first chain of \eqref{eq:degrees} holds
  with equality.
\item[{$L_d$, $\mathcal{L}$, $\mathcal{J}_N$:}] every dyad carries exactly its two
  orientations, so $L_d\equiv2$, $\mathcal{L}=2$ and $\mathcal{J}_N=2N$.
\item[{$\mathbf{1}_{nm}$, $\mathcal{A}$:}] $\mathbf{1}_{nm}=1$ if and only if the two rows
  share a country, whatever their orientations, and in particular
  $\mathbf{1}_{(g,h),(h,g)}=1$. A row on $\{g,h\}$ meets $M_g+M_h-2$ rows, so
  $|\mathcal{A}|=N(4G-6)$ exactly, against the bound
  $N(2\mathcal{M}^{H}-1)=N(4G-5)$ of Lemma \ref{lem:struct}(b).
\item[{$y_n$, $\mathbf{x}_n$:}] the log flow from $g$ to $h$, and the stacked blocks above.
  Here $\|\mathbf{x}_n\|^2=\|\mathbf{z}_g\|^2+\|\mathbf{z}_h\|^2+
  \|\mathbf{s}_{\{g,h\}}\|^2$, so the bound on $\mathbf{x}_n$ in Assumption (A2) holds if
  and only if the attributes are bounded.
\item[{$\boldsymbol{\beta}$, $u_n$:}] the push, pull and pair blocks
  $(\boldsymbol{\beta}_O',\boldsymbol{\beta}_D',\boldsymbol{\beta}_S')'$, and the
  flow-specific disturbance. Note that $u_{gh}$ and $u_{hg}$ are different random
  variables, and that their correlation is the \emph{reciprocity} of the corridor.
\item[{$Q_N$:}] $\Pi$-invariant, $\Pi Q_N\Pi=Q_N$. The invariance is a property of the
  complete directed design rather than an assumption, since every dyad is observed in both
  orientations. The blocks are
  $(1,1)=(2,2)=A_z:=G^{-1}\sum_g\mathbb{E}[\mathbf{z}_g\mathbf{z}_g']$ and
  $(1,2)=(2,1)=B_z:=N^{-1}\sum_{g\ne h}\mathbb{E}[\mathbf{z}_g\mathbf{z}_h']$, with the two
  $\mathbf{z}$--$\mathbf{s}$ blocks equal by the same symmetry. Consequently
  $\mathcal{P}_-$ is $Q_N$-invariant and $Q_N|_{\mathcal{P}_-}=A_z-B_z$. On the tested
  directions Assumption (A3) says that country attributes vary: with $\mathbf{z}_g$
  independent across countries, $A_z-B_z\to\operatorname{Var}(\mathbf{z}_g)$.
\item[{$\widehat{\boldsymbol{\beta}}$, $S_N$:}] one least-squares regression on the $N$
  rows, with score $S_N=\sum_{g\ne h}\mathbf{x}_{gh}u_{gh}$.
\item[{$\Omega_N$, $\widehat{\Omega}_N$:}] $\Omega_N$ contains the reciprocity covariances
  $\operatorname{Cov}(\mathbf{x}_{gh}u_{gh},\mathbf{x}_{hg}u_{hg})$. In
  $\widehat{\Omega}_N$ the $(g,h),(h,g)$ terms are what make the estimator dyadic-robust
  rather than one-way clustered.
\item[{$R_N$, $q$:}] push minus pull, a basis of $\mathcal{P}_-$, constant in $N$, with
  $q=K_1$.
\item[{$\mathcal{V}_N$, $\widehat{\mathcal{V}}_N$:}] the variance of
  $R_N'\widehat{\boldsymbol{\beta}}=\widehat{\boldsymbol{\beta}}_O-
  \widehat{\boldsymbol{\beta}}_D$, and its estimator. The $\Pi$-invariance gives
  $Q_NR_N=R_N(A_z-B_z)$, hence $Q_N^{-1}R_N=R_N(A_z-B_z)^{-1}$ and
  $\mathcal{V}_N=N^{-2}(A_z-B_z)^{-1}R_N'\Omega_NR_N(A_z-B_z)^{-1}$. That is, the bread acts on the
  tested directions as the single $K_1\times K_1$ matrix $A_z-B_z$.
\item[{$W_N$:}] the push-equals-pull Wald statistic, asymptotically $\chi^2_{K_1}$.
\end{description}

Assumption (A1) imposes no restrictions on observations that share a common unit. Consequently, the reciprocity $\operatorname{Corr}(u_{gh},u_{hg})$, as well as the dependence among a country's outflows, its inflows, and the cross-correlation between them, remain unrestricted. Assumption (A1) solely requires that any shock common to node-disjoint corridors be explicitly modeled in the conditional mean.

To evaluate Assumption (A4), note that $R_N'\mathbf{x}_{gh}=\mathbf{z}_g-\mathbf{z}_h$. Therefore, along the tested directions, the score isolates net flows, such that $R_N'S_N=\sum_{\{g,h\}}(\mathbf{z}_g-\mathbf{z}_h)(u_{gh}-u_{hg})$. Consider an error-components model given by $u_{gh}=a_g+b_h+c_{\{g,h\}}+\varepsilon_{gh}$, where $a_g$ denotes a push effect, $b_h$ a pull effect, $c_{\{g,h\}}$ a symmetric corridor shock, and $\varepsilon_{gh}$ is an idiosyncratic noise term. The symmetric corridor shock cancels exactly from the difference $u_{gh}-u_{hg}$, yielding $\operatorname{Var}(R_N'S_N)\asymp G^3\operatorname{Var}(a_g-b_g)+G^2\sigma_\varepsilon^2$. In contrast, the symmetric directions accumulate at order $G^3$ due to the node and corridor shocks. Because $N(\mathcal{M}^{H})^3\asymp8G^5$, Assumption (A4) requires a score variance of order $G^3$ along the tested directions; this condition is satisfied whenever $\operatorname{Var}(a_g-b_g)$ is strictly bounded away from zero.

However, this alone is insufficient to satisfy Assumption (A4) globally, as the condition constrains every direction of $\mathbb{R}^{K}$ and thus requires a strictly positive lower bound on the entire matrix $\Omega_N$. This distinction is theoretically substantive. For instance, setting $b_g\equiv0$ leaves $\operatorname{Var}(a_g-b_g)$ bounded away from zero but eliminates the pull family entirely. In this scenario, the unit aggregates of the remaining push family yield a direction, namely the destination block net of the constant, along which $\Omega_N$ accumulates at order $N$ alone. Consequently, $\lambda_{\min}(\Omega_N)=O(G^{2})$ and $\delta_N$ operates at order $G$.

The theoretical argument necessitates that the unit aggregates vary at order $\mathcal{M}^{H}$ along every direction of $\mathbb{R}^{K}$. This requires the simultaneous presence of both node families; spanning alone is insufficient, as the regressors can span the parameter space while a specific direction still accumulates only at order $N$. A divergence between push and pull effects ensures this requisite variation. In particular, a higher degree of reciprocity in the flows diminishes $\Omega_N$ along the tested directions, while leaving the symmetric directions unaffected. Consequently, the condition number $\kappa(\Omega_N)$ diverges for economically interpretable reasons. Finally, because $\mathcal{J}_N=2N$ operates at order $G^2$, whereas Assumption (A4) restricts $\lambda_{\min}(\Omega_N)$ to grow at a rate strictly faster than $G^{5/2}$, Assumption (A6) holds unconditionally whenever Assumption (A4) is satisfied.

\subsection{Illustration 2: Multilayer Dyads}
\label{sec:smb-multilayer}

Let there be $G$ countries, indexed $g=1,\dots,G$, and $L$ \emph{layers}, say product
categories, indexed $\ell=1,\dots,L$, with $L$ fixed. Write $d=\{g,h\}$ for an unordered
country pair, and let a row of the regression be a (pair, layer) couple $n=(d,\ell)$, every
pair being observed in every layer, so
\[
N=L\binom G2,\qquad \psi\big((d,\ell)\big):=d .
\]
The layer is carried by the row and forgotten by $\psi$, which is $L$-to-one. Again the
non-injectivity is the only feature of the design the framework is told.

Let $y_n=y_{(d,\ell)}$ be the log trade flow of pair $d$ in category $\ell$, and let
$\mathbf{w}_{d\ell}\in\mathbb{R}^{K_1}$ collect the $K_1$ gravity covariates of that pair
in that category, with $K_1$ fixed. These covariates are node-borne by construction, and it
is gravity theory itself that supplies them in that form. Writing
$\boldsymbol{\varphi}_{g\ell}$ for country $g$'s attributes in category $\ell$, the canonical
specification has
$\mathbf{w}_{d\ell}=\boldsymbol{\varphi}_{g\ell}+\boldsymbol{\varphi}_{h\ell}$ plus a bilateral
component, since $\ln(Y_gY_h)=\ln Y_g+\ln Y_h$, and the multilateral resistance terms of
\citet{anderson2003gravity} are node-level as well. Let the regressor be the layer dummy
interacted with the covariates,
\[
\mathbf{x}_n=\mathbf{x}_{(d,\ell)}:=\mathbf{e}_\ell\otimes\mathbf{w}_{d\ell}
\in\mathbb{R}^{K},\qquad K=LK_1,
\]
with $\mathbf{e}_\ell\in\mathbb{R}^{L}$ the $\ell$-th standard basis vector, and
correspondingly $\boldsymbol{\beta}=(\boldsymbol{\beta}_1',\dots,
\boldsymbol{\beta}_L')'$ with $\boldsymbol{\beta}_\ell\in\mathbb{R}^{K_1}$. The single
equation $y_n=\boldsymbol{\beta}'\mathbf{x}_n+u_n$ \emph{is} the system of $L$
category-specific gravity equations
$y_{(d,\ell)}=\boldsymbol{\beta}_\ell'\mathbf{w}_{d\ell}+u_{(d,\ell)}$, which is equation
\eqref{eq:gravityml} of \ref{sec:numerical} below when $\mathbf{w}_{d\ell}$ holds the
pair's size and a constant, and
$\widehat{\boldsymbol{\beta}}$ is exactly equation-by-equation least squares, because
$\sum_n\mathbf{x}_n\mathbf{x}_n'$ is block diagonal. Take the first covariate to be the
pair's economic size, $\ln(Y_gY_h)$, and let the hypothesis be that its elasticity is
common across categories. With $\mathbf{a}_1\in\mathbb{R}^{K_1}$ the first standard basis
vector,
\[
R_N:=\big[\,\mathbf{e}_1\otimes\mathbf{a}_1-\mathbf{e}_2\otimes\mathbf{a}_1\;\big|\;
\cdots\;\big|\;\mathbf{e}_{L-1}\otimes\mathbf{a}_1-\mathbf{e}_L\otimes\mathbf{a}_1\,\big]
\in\mathbb{R}^{K\times q},\qquad q=L-1,\qquad \mathbf{r}=\mathbf{0} .
\]

The dictionary:

\begin{description}
\item[{$g$:}] a country, $g=1,\dots,G$.
\item[{$n$, $N$:}] a (pair, layer) couple $(d,\ell)$, and $N=L\binom G2$.
\item[{$\psi(n)$:}] the country pair $d$, the layer being discarded, so $\psi$ is
  $L$-to-one.
\item[{$M_g$, $\mathcal{M}^{H}$, $\mathcal{M}^{L}$:}] $M_g=L(G-1)$, namely $G-1$ partners
  times $L$ layers, the same for every country, so
  $\mathcal{M}^{H}=\mathcal{M}^{L}=L(G-1)$.
\item[{$L_d$, $\mathcal{L}$, $\mathcal{J}_N$:}] $L_d\equiv L$, $\mathcal{L}=L$ and
  $\mathcal{J}_N=LN$.
\item[{$\mathbf{1}_{nm}$, $\mathcal{A}$:}] equal to one if and only if
  $d\cap d'\ne\emptyset$, whatever the layers. A row on $\{g,h\}$ meets $M_g+M_h-L$ rows,
  so $|\mathcal{A}|=N\cdot L(2G-3)$ exactly.
\item[{$y_n$, $\mathbf{x}_n$:}] the log flow of pair $d$ in category $\ell$, and
  $\mathbf{e}_\ell\otimes\mathbf{w}_{d\ell}$. Here
  $\|\mathbf{x}_n\|=\|\mathbf{w}_{d\ell}\|$, so the bound on $\mathbf{x}_n$ in Assumption
  (A2) is exactly a bound on the covariates.
\item[{$\boldsymbol{\beta}$, $u_n$:}] the $L$ stacked category-specific slope vectors, and
  the pair-by-category disturbance.
\item[{$Q_N$:}] block diagonal, one block per category,
  $Q_N=\operatorname{blockdiag}\big(L^{-1}\Sigma_{N,1},\dots,L^{-1}\Sigma_{N,L}\big)$ with
  $\Sigma_{N,\ell}:=\binom G2^{-1}\sum_d\mathbb{E}[\mathbf{w}_{d\ell}
  \mathbf{w}_{d\ell}']$, so that Assumption (A3) becomes
  $\min_\ell\lambda_{\min}(\Sigma_{N,\ell})\ge Lc_Q$.
\item[{$\widehat{\boldsymbol{\beta}}$, $S_N$:}]
  $\widehat{\boldsymbol{\beta}}_\ell=\big(\sum_d\mathbf{w}_{d\ell}
  \mathbf{w}_{d\ell}'\big)^{-1}\sum_d\mathbf{w}_{d\ell}y_{(d,\ell)}$, category by
  category, while the score stacks the per-category scores
  $S_{N,\ell}=\sum_d\mathbf{w}_{d\ell}u_{(d,\ell)}\in\mathbb{R}^{K_1}$.
\item[{$\Omega_N$, $\widehat{\Omega}_N$:}] an $L\times L$ array of $K_1\times K_1$ blocks,
  with block $(\ell,\ell')=\operatorname{Cov}(S_{N,\ell},S_{N,\ell'})$. The diagonal blocks
  are the per-category score variances and the off-diagonal blocks the cross-category
  covariances, which are left unrestricted. The $(\ell,\ell')$ block of
  $\widehat{\Omega}_N$ is $\sum_{d\cap d'\ne\emptyset}\widehat{u}_{(d,\ell)}
  \widehat{u}_{(d',\ell')}\mathbf{w}_{d\ell}\mathbf{w}_{d'\ell'}'$.
\item[{$R_N$, $q$:}] the $L-1$ size-elasticity contrasts, constant in $N$, with
  $q=L-1$.
\item[{$\mathcal{V}_N$, $\widehat{\mathcal{V}}_N$, $W_N$:}] the $(L-1)\times(L-1)$ variance of those contrasts
  and its estimator, and the cross-category Wald statistic, asymptotically $\chi^2_{L-1}$.
\end{description}

Assumption (A1) leaves the correlation of a single pair's disturbances across categories unrestricted, analogous to the dependence structure in a seemingly unrelated regressions (SUR) system, as well as the correlation between any two pairs sharing a country across any two categories. The assumption strictly requires that a global category shock simultaneously affecting node-disjoint pairs be explicitly observed and incorporated into $\mathbf{x}_n$, or that the analysis be interpreted conditional on its realization. While the category indicator within $\mathbf{x}_n$ absorbs the shock's global mean, it does not absorb the shock itself.

To evaluate Assumption (A4), note that $N(\mathcal{M}^{H})^3\asymp(L^4/2)G^5$, which implies the condition requires $\lambda_{\min}(\Omega_N)\gg (L^2/\sqrt2)G^{5/2}$. Suppose the disturbance for category $\ell$ includes country-specific effects, such that $u_{(d,\ell)}=a_{g\ell}+a_{h\ell}+\varepsilon_{d\ell}$, where $\operatorname{Var}(a_{g\ell})\asymp G^{-\gamma_\ell}$. Because the covariates $\mathbf{w}_{d\ell}$ are node-borne, the unit aggregate $\sum_{h\ne g}\mathbf{w}_{\{g,h\}\ell}$ operates at order $G$ rather than $\sqrt{G}$. Consequently, the score variance for category $\ell$ is of order $G^{3-\gamma_\ell}$, and Assumption (A4) requires $\gamma_\ell<1/2$ for all $\ell$. When cross-category blocks are non-zero, this restriction applies to all directions in $\mathbb{R}^{K}$, not merely to individual categories. The categories may thus exist across a spectrum of genuinely different scales, permitting the condition number $\kappa(\Omega_N)$ to diverge at the rate $G^{\gamma_{\max}-\gamma_{\min}}$. Importantly, the scale parameters $\gamma_\ell$ do not appear in Theorem \ref{thm:main}, nor do they need to be known or estimated. The same restriction matrix $R_N$, variance estimator $\widehat{V}_N$, and $\chi^2_{L-1}$ critical value remain asymptotically valid regardless of whether the categories share a common scale, a property explicitly evaluated in the multilayer Monte Carlo experiments of Appendix \ref{sec:mc2}.

The requirement that covariates be node-borne is important, imposing a restriction on the empirical design rather than on the theorem itself. A covariate with idiosyncratic pair-level variation, defined such that its $G-1$ values for a given country are uncorrelated and lack a country-specific mean, contributes a direction to $\mathbb{R}^{K}$ where the unit aggregate is of order $\sqrt{G}$ rather than $G$, provided the intercept absorbs any corresponding global level. Whether a covariate exhibits this idiosyncratic property depends on its fundamental nature, not simply its bilateral structure. An independent, mean-zero pair shock satisfies this condition; however, log distance does not for example. For log distance, the aggregate $\sum_{h\ne g}(d_{\{g,h\}}-\bar{d})=(G-1)(\bar{d}_g-\bar{d})$ is of order $G$ for any country whose average distance diverges from the overall sample average (e.g., geographically remote countries). Projecting onto category intercepts fails to resolve this, as these intercepts are global per-category constants that cannot remove country-specific means.

Along a direction defined by such an idiosyncratic covariate, $\Omega_N$ accumulates only at order $N$. Because Assumption (A4) constrains every direction in the parameter space rather than just the tested contrast, the condition will fail irrespective of how well-scaled the tested contrast is. Accommodating an idiosyncratic covariate within this design therefore requires a support sufficiently sparse such that the lower bound $\{N(\mathcal{M}^{H})^{3}\}^{1/2}$ falls below order $N$, implying $(\mathcal{M}^{H})^{3}=o(N)$. Finally, because $\mathcal{J}_N=LN$ is of order $G^2$, Assumption (A6) is satisfied unconditionally whenever Assumption (A4) holds. 

\subsection{Illustration 3: A Balanced Dyadic Panel}
\label{sec:smb-panel}

Let there be $G$ airports, indexed $g=1,\dots,G$, and $T$ years, indexed $t=1,\dots,T$,
with $T$ fixed. Write $d=\{g,h\}$ for a city-pair \emph{route}, and let a row of the
regression be a (route, year) couple $n=(d,t)$, every route being observed in every year,
so
\[
N=T\binom G2,\qquad \psi\big((d,t)\big):=d ,
\]
the year being carried by the row and forgotten by $\psi$, which is $T$-to-one.

Let $y_n=y_{dt}$ be log passengers flown on route $d$ in year $t$. Let
$\mathbf{s}_d\in\mathbb{R}^{K_s}$ collect time-invariant route characteristics and
$\mathbf{f}_{dt}\in\mathbb{R}^{K_w}$ time-varying ones, say fares. Both are node-borne,
with $\mathbf{s}_d=\boldsymbol{\theta}_g+\boldsymbol{\theta}_h$ plus a bilateral component
such as great-circle distance, and
$\mathbf{f}_{dt}=\boldsymbol{\varphi}_{gt}+\boldsymbol{\varphi}_{ht}$, where
$\boldsymbol{\theta}_g$ is a permanent airport characteristic and
$\boldsymbol{\varphi}_{gt}$ a time-varying one. Split the time-varying block into its route
mean and its deviation,
\[
\bar{\mathbf{f}}_d:=\frac1T\sum_{r=1}^{T}\mathbf{f}_{dr},\qquad
\mathbf{f}^{\mathrm w}_{dt}:=\mathbf{f}_{dt}-\bar{\mathbf{f}}_d,
\qquad\text{so that }\ \sum_t\mathbf{f}^{\mathrm w}_{dt}=\mathbf{0}\ \text{ exactly for
every }d,
\]
and let the regressor stack the year dummies, the time-invariant block, the
\emph{between} block, and the \emph{within} block,
\[
\mathbf{x}_n=\mathbf{x}_{dt}:=\big(\mathbf{e}_t',\ \mathbf{s}_d',\
\bar{\mathbf{f}}_d',\ (\mathbf{f}^{\mathrm w}_{dt})'\big)'\in\mathbb{R}^{K},
\qquad K=T+K_s+2K_w,
\]
with $\mathbf{e}_t\in\mathbb{R}^{T}$ the $t$-th standard basis vector, so that the
specification carries a full set of year dummies and no separate intercept. With
$\boldsymbol{\beta}=(\boldsymbol{\delta}',
\boldsymbol{\beta}_s',\boldsymbol{\beta}_B',\boldsymbol{\beta}_W')'$ the equation reads
$y_{dt}=\delta_t+\boldsymbol{\beta}_s'\mathbf{s}_d+\boldsymbol{\beta}_B'
\bar{\mathbf{f}}_d+\boldsymbol{\beta}_W'\mathbf{f}^{\mathrm w}_{dt}+u_{dt}$. The year
dummies are what fixes $T$. Assumption (A1) asks that a shock common to node-disjoint
routes, such as an aggregate demand or fuel-price shock, sit in the mean, and that costs
$T$ parameters; since Assumption (A5) fixes $K$, we take $T$ fixed and $G\to\infty$, a
short panel of many units. Let the hypothesis be the between--within contrast of
\citet{mundlak1978pooling} and \citet{hausman1978specification}, namely that between-route
and within-route variation deliver the same coefficient,
$\boldsymbol{\beta}_B=\boldsymbol{\beta}_W$,
\[
R_N:=\big(0_{K_w\times T},\ 0_{K_w\times K_s},\ I_{K_w},\ -I_{K_w}\big)'
\in\mathbb{R}^{K\times q},\qquad q=K_w,\qquad \mathbf{r}=\mathbf{0} ,
\]
a dyadic-robust Hausman test with no rate assumed anywhere.

The dictionary:

\begin{description}
\item[{$g$:}] an airport, $g=1,\dots,G$.
\item[{$n$, $N$:}] a (route, year) couple $(d,t)$, and $N=T\binom G2$.
\item[{$\psi(n)$:}] the route, the year being discarded, so $\psi$ is $T$-to-one.
\item[{$M_g$, $\mathcal{M}^{H}$, $\mathcal{M}^{L}$:}] $M_g=T(G-1)$, namely $G-1$ routes
  times $T$ years, the same for every airport, so
  $\mathcal{M}^{H}=\mathcal{M}^{L}=T(G-1)$. The equality is what a balanced panel buys.
\item[{$L_d$, $\mathcal{L}$, $\mathcal{J}_N$:}] $L_d\equiv T$, $\mathcal{L}=T$ and
  $\mathcal{J}_N=TN$.
\item[{$\mathbf{1}_{nm}$, $\mathcal{A}$:}] equal to one if and only if the two rows share
  an airport, whatever the years, and in particular for the same route at any two dates.
  Here $|\mathcal{A}|=N\cdot T(2G-3)$ exactly.
\item[{$y_n$, $\mathbf{x}_n$:}] log passengers on route $d$ in year $t$, and the stack
  above. Here $\|\mathbf{x}_n\|^2=1+\|\mathbf{s}_d\|^2+\|\bar{\mathbf{f}}_d\|^2+
  \|\mathbf{f}^{\mathrm w}_{dt}\|^2$, so the bound on $\mathbf{x}_n$ in Assumption (A2)
  holds if and only if the characteristics are bounded.
\item[{$\boldsymbol{\beta}$, $u_n$:}] $(\boldsymbol{\delta}',\boldsymbol{\beta}_s',
  \boldsymbol{\beta}_B',\boldsymbol{\beta}_W')'$, one common slope pair rather than $T$ of
  them, and the route-year disturbance. Note that $u_{d1},\dots,u_{dT}$ are $T$ different
  random variables on one route.
\item[{$Q_N$:}] the between and within blocks are orthogonal by construction, as the within block is route-demeaned, yielding exactly $N^{-1}\sum_{d,t}\bar{\mathbf{f}}_d(\mathbf{f}^{\mathrm w}_{dt})'= N^{-1}\sum_d\bar{\mathbf{f}}_d\big(\sum_t\mathbf{f}^{\mathrm w}_{dt}\big)'=0$. Consequently, $Q_N$ restricted to the $(B,W)$ subspace is block diagonal $\operatorname{diag}(\Sigma_B,\Sigma_W)$, where $\Sigma_B:=\binom G2^{-1}\sum_d\mathbb{E}[\bar{\mathbf{f}}_d\bar{\mathbf{f}}_d']$ and $\Sigma_W:=N^{-1}\sum_{d,t}\mathbb{E}[\mathbf{f}^{\mathrm w}_{dt} (\mathbf{f}^{\mathrm w}_{dt})']$. However, the year-dummy and $\mathbf{s}_d$ blocks are not orthogonal to this subspace, necessitating that the restriction $R_N$ partial them out. Along the tested directions, Assumption (A3) requires both cross-sectional and intertemporal variation: $\Sigma_B\succ0$ establishes the full rank of the raw between-block second moment, which does not strictly measure between-route variation, as it holds even if all routes share an identical nonzero average fare; $\Sigma_W\succ0$ implies that fares exhibit temporal variation within a given route. These conditions are necessary but not sufficient; Assumption (A3) additionally requires that both blocks remain nonsingular after partialling out the year dummies and $\mathbf{s}_d$, which formally ensures between-route identification.
\item[{$\widehat{\boldsymbol{\beta}}$, $S_N$:}] one pooled least-squares regression on the
  $N$ rows, with score $S_N=\sum_{d,t}\mathbf{x}_{dt}u_{dt}$.
\item[{$\Omega_N$, $\widehat{\Omega}_N$:}] $\Omega_N$ contains the serial covariances
  $\operatorname{Cov}(\mathbf{x}_{dt}u_{dt},\mathbf{x}_{ds}u_{ds})$ at every lag. In
  $\widehat{\Omega}_N$ the within-route pairs enter with weight one at every lag, with no
  kernel, no bandwidth and no lag truncation, since $T$ is fixed.
\item[{$R_N$, $q$:}] the between--within contrast
  $\boldsymbol{\beta}_B-\boldsymbol{\beta}_W$, constant in $N$, with $q=K_w$.
\item[{$\mathcal{V}_N$, $\widehat{\mathcal{V}}_N$, $W_N$:}] the $K_w\times K_w$ variance of
  $\widehat{\boldsymbol{\beta}}_B-\widehat{\boldsymbol{\beta}}_W$ and its estimator, and
  the dyadic-robust Hausman statistic, asymptotically $\chi^2_{K_w}$.
\end{description}

Assumption (A1) permits arbitrary serial dependence within a given route, imposing no stationarity, mixing, or bandwidth conditions, as well as arbitrary dependence between routes sharing an airport across any time periods. The assumption strictly requires that shocks common to node-disjoint routes within a year be explicitly captured by the conditional mean. Year indicators satisfy this by treating time effects as fixed parameters, effectively conditioning the analysis on their realization; however, heterogeneous shock loadings across routes would not be absorbed in this manner.

To evaluate Assumption (A4), observe that $N(\mathcal{M}^{H})^3\asymp(T^4/2)G^5$, which imposes the requirement $\lambda_{\min}(\Omega_N)\gg(T^2/\sqrt2)G^{5/2}$. The within-route directions isolate strictly temporal variation, as $\sum_t\mathbf{f}^{\mathrm w}_{dt}=\mathbf{0}$ exactly. Suppose every direction in $\mathbb{R}^{K}$ exhibits unit aggregates of order $G$. This occurs with node-borne components such as $\boldsymbol{\theta}_g+\boldsymbol{\theta}_h$ and $\boldsymbol{\varphi}_{gt}+\boldsymbol{\varphi}_{ht}$, as well as with great-circle distance through its country-specific mean (as discussed in Appendix \ref{sec:smb-multilayer}). Under this structure, between-type directions accumulate at order $T^2G^3$. The within directions surpass the $G^{5/2}$ threshold provided that the airport-borne shocks are time-varying and the covariates $\mathbf{f}_{dt}$ are inherently node-borne (as would be typical in a fare panel). Consequently, Assumption (A4) is satisfied; furthermore, because $\mathcal{J}_N=TN$ operates at order $G^2$, it is strictly dominated by $\lambda_{\min}(\Omega_N)$, ensuring that Assumption (A6) holds simultaneously.

However, panel balancedness and temporal variation are insufficient to guarantee this lower bound inherently. It is possible to construct a bounded bilateral route characteristic satisfying $\sum_{h\ne g}z_{\{g,h\}}=0$ for all $g$ while maintaining $\sum_dz_d^2\asymp G^2$. Such a covariate perfectly eliminates airport shocks from its corresponding score coordinate, reducing the variance accumulation along that specific direction to order $N$. In this scenario, $\lambda_{\min}(\Omega_N)=\Theta(N)$, $\delta_N=\Theta(G)$, and $\mathcal{J}_N/\lambda_{\min}(\Omega_N)=\Theta(T)$, meaning neither asymptotic condition vanishes. Ultimately, satisfying these assumptions relies on establishing a primitive lower bound on $\lambda_{\min}(\Omega_N)$, rather than depending solely on the geometric structure of the panel.

\section{Numerical Exercise}
\label{sec:numerical}

In this section, we apply both theorems in the main text to a bilateral trade gravity equation across three distinct dyadic configurations. Two of the three treat an observation as a directed country pair:
\begin{equation}\label{eq:gravity}
\ln T_{ij}\;=\;\underbrace{\boldsymbol\beta_O'\mathbf z_i}_{\text{push, exporter}}
+\underbrace{\boldsymbol\beta_D'\mathbf z_j}_{\text{pull, importer}}
+\boldsymbol\beta_S'\mathbf s_{\{i,j\}}+u_{ij},
\end{equation}
where $T_{ij}$ is the directed trade flow from exporter $i$ to importer $j$, $\mathbf{z}_i$ and $\mathbf{z}_j$ represent origin and destination attributes, respectively, and $\mathbf{s}_{\{i,j\}}$ captures symmetric pair attributes, including the constant. Under this structure, the mapping $\psi\big((i,j)\big)=\{i,j\}$ is at most two-to-one, achieving exactly two-to-one solely for reciprocated pairs. Consequently, a realized dyad carries either one or two observations, which aligns with Illustration 1 in \ref{sec:smb}.

The remaining configuration defines an observation as an \emph{unordered} country pair measured across $L$ product categories:
\begin{equation}\label{eq:gravityml}
\ln T_{d\ell}\;=\;\beta_{\ell,1}\,w_{d\ell}+\beta_{\ell,2}+u_{d\ell},
\qquad \ell=1,\dots,L,
\end{equation}
where $d=\{g,h\}$ denotes an unordered pair, $w_{d\ell}=\ln(Y_gY_h)=\ln Y_g+\ln Y_h$ is the pair's economic size in category $\ell$, and $\beta_{\ell,2}$ serves as the category intercept. In this case, $\psi\big((d,\ell)\big)=d$ constitutes an $L$-to-one mapping, corresponding directly to Illustration 2 in \ref{sec:smb}. The two provided dictionaries map every symbol to the established notation in Section \ref{sec:model} of the main text.

By retaining a single gravity relationship across three non-injective configurations of the dyadic index set, these exercises maintain a unified economic interpretation while allowing the configuration, regressor blocks, and error structures to vary. Appendix \ref{sec:mc} simulates the model in \eqref{eq:gravity} on directed supports of varying density to evaluate the performance of $\widehat V_N$ and $\widehat V^{J}$ relative to standard alternative variance estimators. Appendix \ref{sec:mc2} simulates the specification in \eqref{eq:gravityml} over a multilayer support to assess joint inference when linear combinations of $\widehat{\boldsymbol\beta}$ converge at heterogeneous rates. Finally, Appendix \ref{sec:empirical} estimates the model in \eqref{eq:gravity} using aggregate cross-sectional trade data. The exercises of Appendices \ref{sec:mc} and \ref{sec:empirical} are structurally linked through their shared specification and a simulated density explicitly calibrated to the estimation sample. All computational procedures are executed in \texttt{Python} relying on the \texttt{numpy} package \citep{harris2020array}.

\subsection{Monte Carlo Experiments: Directed Dyads}
\label{sec:mc}

The simulations specify $\mathbf{z}_i=(z_{1i},z_{2i})'$ and a single symmetric attribute plus a constant in $\mathbf{s}_{\{i,j\}}$, establishing $K=6$. The true parameters are $\boldsymbol\beta_{O}=\boldsymbol\beta_{D}=(1,1/2)'$ and $\boldsymbol\beta_{S}=(1,0)'$. The support is a directed Erd\H{o}s--R\'{e}nyi graph of $G\in\{40,80,160\}$ nodes, where each ordered pair $(i,j)$, $i\ne j$, is observed independently with probability $p$. The draw is repeated until no node is isolated, ensuring $\mathcal{M}^{L}\ge1$. We consider three density rules: $p=\log G/G$, $p=2/\sqrt G$, and $p=0.52$. The latter matches the share of non-zero directed trade flows in the empirical application of Appendix \ref{sec:empirical}. These densities span regimes where Assumption (A4) is respectively automatic, knife-edge, and binding. Node attributes $z_{1g},z_{2g},\vartheta_{g}$ are independent $\mathrm{U}[-1,1]$ with $s_{\{i,j\}}=(\vartheta_i+\vartheta_j)/2$. The errors follow $u_{ij}=a_i+b_j+c_{\{i,j\}}+e_{ij}$, where $a_g=\sigma_a(\mu_g+\nu^{a}_g)/\sqrt2$ and $b_g=\sigma_b(\mu_g+\nu^{b}_g)/\sqrt2$. The primitives $\mu_g,\nu^{a}_g,\nu^{b}_g,c_{\{i,j\}},e_{ij}$ are independent $\mathrm{U}[-1,1]$, and we set $\sigma_a=1$ and $\sigma_b=3/5$. This structure guarantees $\operatorname{Corr}(a_g,b_g)=1/2$, enforces bounded components as required by Assumption (A2), and permits unrestricted reciprocity alongside covariance between a node's outflows and inflows. We execute 5{,}000 replications per configuration, redrawing the graph, attributes, and errors in each. Bias is omitted as the estimator is unbiased by construction.

Table \ref{tab:mc_se_coverage} details the ratios of the average estimated standard errors relative to the Monte Carlo standard deviation, alongside the empirical coverage of the nominal $95\%$ confidence interval for $\widehat{V}_N$. Table \ref{tab:mc_size_power} reports both the raw size and the size-adjusted local power for testing the null hypothesis $H_0:\boldsymbol\beta_{O}=\boldsymbol\beta_{D}$ ($q=2$), providing comparisons against i.i.d., one-way, and two-way \citep{cameron2011robust} variance estimators.

The local alternatives are indexed by a reference noncentrality parameter $\lambda$, parameterized as $\boldsymbol\beta_{O}-\boldsymbol\beta_{D}=\bar{\mathcal{V}}_N^{1/2}\mathbf{h}$. Here, $\mathbf{h}=(\lambda/q)^{1/2}\mathbf{1}_q$ represents the equiangular direction ensuring $\|\mathbf{h}\|^2=\lambda$, and $\bar{\mathcal{V}}_N$ is defined as the average of $R_N'V_NR_N$ computed across a preliminary set of configuration draws. The actual noncentrality realized in any specific draw is given by $\mathbf{h}'\bar{\mathcal{V}}_N^{1/2}\mathcal{V}_N^{-1}\bar{\mathcal{V}}_N^{1/2}\mathbf{h}$, which precisely equals $\lambda$ only when $\mathcal{V}_N$ exactly coincides with $\bar{\mathcal{V}}_N$. Thus, $\lambda$ strictly indexes the sequence of alternatives rather than the realized noncentrality of any individual configuration. This indexing approach deliberately isolates statistical power from estimation precision, acknowledging that the Monte Carlo standard deviation of $\widehat\beta_{O,1}$ decreases substantially as $G$ increases.

Reported size is raw, whereas power is strictly size-adjusted; each test statistic is evaluated against its own empirical $95$th percentile under the null ($\lambda=0$). Furthermore, neither of the cluster comparators incorporates a finite-sample correction. Following the convention established in Section \ref{sec:assumptions}, Wald statistics are explicitly set to zero whenever the inverted variance matrix is not positive definite. In these simulations, this singularity occurs in $0.17\%$ of replications for $R_N'\widehat{V}_NR_N$, but never for $R_N'\widehat{V}^{J}R_N$, which is positive semidefinite by construction.

\begin{table}[tp]
\centering
\spacingset{1}
\setlength{\tabcolsep}{3pt}
\footnotesize
\caption{Standard errors and coverage across nine directed Erd\H{o}s--R\'enyi
configurations: the dyadic-robust estimator, its delete-a-unit jackknife, and the
i.i.d.\ benchmark.}
\label{tab:mc_se_coverage}
\begin{tabular}{@{}rrrrr rrrr rrrr rrrr@{}}
\toprule
 & & & & & \multicolumn{4}{c}{$\beta_{O,1}$} & \multicolumn{4}{c}{$\beta_{D,1}$} & \multicolumn{4}{c}{$\beta_{S,1}$} \\
\cmidrule(lr){6-9}\cmidrule(lr){10-13}\cmidrule(lr){14-17}
$G$ & $p$ & $N$ & $\delta_N$ & $\mathcal{J}_N/\lambda_{\min}$
 & $r_{\mathrm{dy}}$ & $r_{\mathrm{jk}}$ & $r_{\mathrm{iid}}$ & Cov
 & $r_{\mathrm{dy}}$ & $r_{\mathrm{jk}}$ & $r_{\mathrm{iid}}$ & Cov
 & $r_{\mathrm{dy}}$ & $r_{\mathrm{jk}}$ & $r_{\mathrm{iid}}$ & Cov \\
\midrule
 40 & 0.092 &       145 & 285.2 & 4.42 & 0.828 & 1.251 & 0.688 & 87.3 & 0.835 & 1.387 & 0.857 & 87.0 & 0.809 & 1.279 & 0.698 & 85.8 \\
 40 & 0.316 &       495 & 366.5 & 2.82 & 0.861 & 1.119 & 0.457 & 89.3 & 0.841 & 1.231 & 0.664 & 87.8 & 0.838 & 1.142 & 0.463 & 88.1 \\
 40 & 0.520 &       810 & 405.7 & 2.43 & 0.874 & 1.091 & 0.373 & 90.0 & 0.829 & 1.153 & 0.556 & 87.3 & 0.840 & 1.102 & 0.375 & 88.7 \\
\addlinespace
 80 & 0.055 &       348 & 116.6 & 3.20 & 0.934 & 1.238 & 0.671 & 92.8 & 0.915 & 1.327 & 0.828 & 91.3 & 0.915 & 1.242 & 0.673 & 91.6 \\
 80 & 0.224 &   1{,}413 & 184.5 & 1.82 & 0.929 & 1.081 & 0.397 & 92.2 & 0.925 & 1.177 & 0.593 & 92.2 & 0.927 & 1.104 & 0.405 & 92.1 \\
 80 & 0.520 &   3{,}287 & 207.3 & 1.30 & 0.951 & 1.059 & 0.277 & 93.5 & 0.921 & 1.094 & 0.432 & 91.8 & 0.933 & 1.066 & 0.279 & 92.6 \\
\addlinespace
160 & 0.032 &       806 &  57.8 & 2.60 & 0.951 & 1.181 & 0.626 & 93.0 & 0.964 & 1.307 & 0.813 & 93.6 & 0.951 & 1.196 & 0.635 & 93.3 \\
160 & 0.158 &   4{,}029 &  96.1 & 1.27 & 0.958 & 1.050 & 0.340 & 93.0 & 0.945 & 1.113 & 0.514 & 93.6 & 0.974 & 1.081 & 0.351 & 94.1 \\
160 & 0.520 &  13{,}233 & 101.6 & 0.68 & 0.977 & 1.031 & 0.199 & 94.5 & 0.956 & 1.045 & 0.318 & 93.5 & 0.957 & 1.022 & 0.198 & 93.4 \\
\bottomrule
\end{tabular}

\vspace{4pt}
\begin{minipage}{\textwidth}\scriptsize
\textsuperscript{a}\,$G$ is the number of nodes and $p$ the edge probability of the
directed Erd\H{o}s--R\'enyi support; $N$ is the mean number of observations across the
5{,}000 replications; $\delta_N=N(\mathcal{M}^{H})^{3}/\lambda_{\min}(\Omega_N)^{2}$ is
the quantity Assumption~(A4) sends to zero and $\mathcal{J}_N/\lambda_{\min}$ abbreviates
$\mathcal{J}_N/\lambda_{\min}(\Omega_N)$, the quantity Assumption~(A6) sends to zero,
both computed exactly from $\Omega_N$ on the realized configuration and averaged over
draws.\\
\textsuperscript{b}\,Writing SD for the Monte Carlo standard deviation of $\widehat\beta$
across replications, the entries are the normalized ratios
$r_{\mathrm{dy}}=\overline{\mathrm{SE}}(\widehat V_N)/\mathrm{SD}$,
$r_{\mathrm{jk}}=\overline{\mathrm{SE}}(\widehat V^{J})/\mathrm{SD}$ and
$r_{\mathrm{iid}}=\overline{\mathrm{SE}}(\text{i.i.d.})/\mathrm{SD}$, each SE an average
over replications; Cov is the empirical coverage of the nominal $95\%$ interval built
from $\widehat V_N$, in per cent.
\end{minipage}
\end{table}

\begin{table}[tp]
\centering
\spacingset{1}
\setlength{\tabcolsep}{4pt}
\footnotesize
\caption{Size and local power of the test of $H_0:\boldsymbol\beta_{O}=\boldsymbol\beta_{D}$ across nine directed Erd\H{o}s--R\'enyi configurations, against three alternative variance estimators.}
\label{tab:mc_size_power}
\begin{tabular}{@{}rrrr rrr rrrr rrrr@{}}
\toprule
 & & & & \multicolumn{3}{c}{Size, alternative estimators} & \multicolumn{4}{c}{$W_N$} & \multicolumn{4}{c}{$W_N^{J}$} \\
\cmidrule(lr){5-7}\cmidrule(lr){8-11}\cmidrule(lr){12-15}
$G$ & $p$ & $\delta_N$ & $\mathcal{J}_N/\lambda_{\min}$
 & i.i.d. & One-way & Two-way
 & Size & $\lambda=2$ & $\lambda=5$ & $\lambda=10$
 & Size & $\lambda=2$ & $\lambda=5$ & $\lambda=10$ \\
\midrule
 40 & 0.092 & 285.2 & 4.42 & 0.121 & 0.095 & 0.093 & 0.198 & 0.148 & 0.318 & 0.559 & 0.014 & 0.221 & 0.462 & 0.756 \\
 40 & 0.316 & 366.5 & 2.82 & 0.273 & 0.060 & 0.040 & 0.144 & 0.187 & 0.397 & 0.693 & 0.031 & 0.205 & 0.447 & 0.755 \\
 40 & 0.520 & 405.7 & 2.43 & 0.388 & 0.059 & 0.031 & 0.132 & 0.182 & 0.400 & 0.695 & 0.047 & 0.203 & 0.441 & 0.737 \\
\addlinespace
 80 & 0.055 & 116.6 & 3.20 & 0.131 & 0.057 & 0.047 & 0.105 & 0.200 & 0.441 & 0.738 & 0.012 & 0.219 & 0.480 & 0.786 \\
 80 & 0.224 & 184.5 & 1.82 & 0.357 & 0.040 & 0.019 & 0.085 & 0.210 & 0.469 & 0.771 & 0.026 & 0.213 & 0.481 & 0.787 \\
 80 & 0.520 & 207.3 & 1.30 & 0.580 & 0.037 & 0.016 & 0.084 & 0.213 & 0.464 & 0.773 & 0.044 & 0.213 & 0.468 & 0.776 \\
\addlinespace
160 & 0.032 &  57.8 & 2.60 & 0.147 & 0.046 & 0.032 & 0.073 & 0.216 & 0.473 & 0.780 & 0.010 & 0.219 & 0.487 & 0.795 \\
160 & 0.158 &  96.1 & 1.27 & 0.453 & 0.029 & 0.011 & 0.066 & 0.222 & 0.493 & 0.798 & 0.025 & 0.228 & 0.500 & 0.809 \\
160 & 0.520 & 101.6 & 0.68 & 0.739 & 0.026 & 0.008 & 0.064 & 0.225 & 0.491 & 0.810 & 0.044 & 0.226 & 0.495 & 0.812 \\
\midrule
\multicolumn{4}{@{}l}{\emph{Limiting value}} & 0.050 & 0.050 & 0.050 & 0.050 & 0.226 & 0.504 & 0.815 & 0.050 & 0.226 & 0.504 & 0.815 \\
\bottomrule
\end{tabular}

\vspace{4pt}
\begin{minipage}{\textwidth}\scriptsize
\textsuperscript{a}\,Entries are rejection frequencies, over 5{,}000 replications per
row, of nominal $5\%$ tests of $H_0:\boldsymbol\beta_{O}=\boldsymbol\beta_{D}$; the
configurations and the columns $G$, $p$, $\delta_N$ and $\mathcal{J}_N/\lambda_{\min}$
are those of Table~\ref{tab:mc_se_coverage}.\\
\textsuperscript{b}\,i.i.d.\ is the Wald test built from $\widehat\sigma^{2}(X'X)^{-1}$;
One-way clusters on the exporter; Two-way is the \citet{cameron2011robust} estimator with
intersection clusters equal to single ordered pairs, that is
$\sum_{g}T^{O}_{g}T^{O\prime}_{g}+\sum_{g}T^{D}_{g}T^{D\prime}_{g}-\sum_{n}\mathbf{x}_n\mathbf{x}_n'\widehat
u_n^{2}$ with $T^{O}_{g}=\sum_{n:\,i(n)=g}\mathbf{x}_n\widehat u_n$ and $T^{D}_{g}$ its
importer counterpart.\\
\textsuperscript{c}\,Size is the raw rejection frequency under the null; the $\lambda$
columns report size-adjusted power against the local alternatives
$\boldsymbol\beta_{O}-\boldsymbol\beta_{D}=\bar M_N^{1/2}\mathbf{h}$ with noncentrality
$\lambda=\lVert\mathbf{h}\rVert^{2}$, each statistic evaluated against its own empirical
$95$th percentile under $\lambda=0$; the Limiting value row is the corresponding
noncentral $\chi^{2}_{2}$ rejection probability.

\end{minipage}
\end{table}

The simulation results yield four primary conclusions. First, the dyadic-robust standard errors display a downward bias in finite samples; as $G$ increases, $r_{\mathrm{dy}}$ approaches one from $0.81$, and empirical coverage improves from $85.8\%$ to $94.5\%$. Conversely, the i.i.d. standard errors are fundamentally inconsistent. At the densest configuration, they capture only a fraction of the true sampling variability, yielding a test size that reaches $73.9\%$ and strictly deteriorates as the sample size increases. Second, one-way clustering on the exporter and two-way clustering \citep{cameron2011robust} err in the opposite direction, with their empirical size dropping as low as $0.008$. Both of these estimators omit the covariance between a node's outflows and inflows. Because this covariance is negative for this specific contrast, and given that the sign is inherently contrast-dependent, neither estimator produces predictable errors across different hypotheses. Third, consistent with the theoretical predictions of Proposition \ref{prop:onesided}, $W^{J}_N$ closely approximates the nominal size when $\mathcal{J}_N/\lambda_{\min}(\Omega_N)$ is small, but exhibits a conservative one-sided error when this ratio is large. The jackknife standard error ratio, $r_{\mathrm{jk}}$, strictly exceeds one and monotonically increases with $\mathcal{J}_N/\lambda_{\min}(\Omega_N)$. Finally, across all 27 evaluated local alternatives, the size-adjusted jackknife test strictly dominates or equals the statistical power of the base statistic.

Table \ref{tab:mc_size_power_q1} replicates this analysis for the single restriction $H_0:\beta_{O,1}=\beta_{D,1}$ ($q=1$). For a single restriction, the base statistic $W_N$ exhibits reduced size distortion, with the raw size in the sparse configuration at $G=40$ decreasing from $0.198$ to $0.122$. Nevertheless, the jackknife remains conservative at high values of $\mathcal{J}_N/\lambda_{\min}(\Omega_N)$ across both testing exercises. Equation \eqref{eq:master} uniformly bounds the normalized error over restrictions, as the restriction enters exclusively through the congruence $R_N^{*}$ of norm one. This uniformity ensures the theoretical comparability of the two exercises; however, it does not imply finite-sample invariance with respect to $q$ or the specific restriction chosen. Consequently, the evaluated sizes at $\mathcal{J}_N/\lambda_{\min}(\Omega_N)=2.43$ systematically differ, yielding $0.047$ for $q=2$ compared to $0.039$ for $q=1$.

\begin{table}[tp]
\centering
\spacingset{1}
\setlength{\tabcolsep}{4pt}
\footnotesize
\caption{Size and local power of the test of the single restriction $H_0:\beta_{O,1}=\beta_{D,1}$, across the same nine configurations.}
\label{tab:mc_size_power_q1}
\begin{tabular}{@{}rrrr rrr rrrr rrrr@{}}
\toprule
 & & & & \multicolumn{3}{c}{Size, alternative estimators} & \multicolumn{4}{c}{$W_N$} & \multicolumn{4}{c}{$W_N^{J}$} \\
\cmidrule(lr){5-7}\cmidrule(lr){8-11}\cmidrule(lr){12-15}
$G$ & $p$ & $\delta_N$ & $\mathcal{J}_N/\lambda_{\min}$
 & i.i.d. & One-way & Two-way
 & Size & $\lambda=2$ & $\lambda=5$ & $\lambda=10$
 & Size & $\lambda=2$ & $\lambda=5$ & $\lambda=10$ \\
\midrule
 40 & 0.092 & 285.2 & 4.42 & 0.094 & 0.072 & 0.068 & 0.122 & 0.233 & 0.506 & 0.789 & 0.013 & 0.280 & 0.576 & 0.850 \\
 40 & 0.316 & 366.5 & 2.82 & 0.194 & 0.047 & 0.032 & 0.095 & 0.262 & 0.545 & 0.831 & 0.026 & 0.290 & 0.585 & 0.861 \\
 40 & 0.520 & 405.7 & 2.43 & 0.265 & 0.051 & 0.028 & 0.090 & 0.261 & 0.545 & 0.833 & 0.039 & 0.274 & 0.565 & 0.846 \\
\addlinespace
 80 & 0.055 & 116.6 & 3.20 & 0.110 & 0.054 & 0.044 & 0.082 & 0.259 & 0.560 & 0.848 & 0.013 & 0.272 & 0.581 & 0.869 \\
 80 & 0.224 & 184.5 & 1.82 & 0.247 & 0.035 & 0.023 & 0.070 & 0.278 & 0.587 & 0.867 & 0.027 & 0.284 & 0.596 & 0.873 \\
 80 & 0.520 & 207.3 & 1.30 & 0.406 & 0.032 & 0.015 & 0.069 & 0.285 & 0.582 & 0.868 & 0.040 & 0.291 & 0.589 & 0.873 \\
\addlinespace
160 & 0.032 &  57.8 & 2.60 & 0.116 & 0.043 & 0.032 & 0.063 & 0.279 & 0.593 & 0.873 & 0.014 & 0.287 & 0.598 & 0.881 \\
160 & 0.158 &  96.1 & 1.27 & 0.311 & 0.034 & 0.015 & 0.060 & 0.287 & 0.594 & 0.879 & 0.031 & 0.286 & 0.595 & 0.881 \\
160 & 0.520 & 101.6 & 0.68 & 0.540 & 0.029 & 0.011 & 0.060 & 0.287 & 0.595 & 0.875 & 0.045 & 0.289 & 0.597 & 0.877 \\
\midrule
\multicolumn{4}{@{}l}{\emph{Limiting value}} & 0.050 & 0.050 & 0.050 & 0.050 & 0.293 & 0.609 & 0.885 & 0.050 & 0.293 & 0.609 & 0.885 \\
\bottomrule
\end{tabular}

\vspace{4pt}
\begin{minipage}{\textwidth}\scriptsize
\textsuperscript{a}\,This corresponds to Table~\ref{tab:mc_size_power}, but where the
hypothesis tested is the single restriction $H_0:\beta_{O,1}=\beta_{D,1}$, so that
$R_N=(1,0,-1,0,0,0)'$ and $q=1$, computed on the same draws and seeds; every column is
defined as in Table~\ref{tab:mc_size_power}, with the limiting values now those of the
noncentral $\chi^{2}_{1}$ at its $95$th percentile 3.841.

\end{minipage}
\end{table}

\subsection{Monte Carlo Experiments: Multilayer Dyads}
\label{sec:mc2}

The second Monte Carlo exercise estimates the specification in \eqref{eq:gravityml} across $L=3$ categories. The regressor is defined as $\mathbf{x}_n=\mathbf{e}_\ell\otimes(w_{d\ell},1)'$ with $K=LK_1=6$, where the true parameters are set to $\beta_{\ell,1}=1$ and $\beta_{\ell,2}=0$ for all $\ell$. The support is generated as an undirected Erd\H{o}s--R\'{e}nyi graph over $G\in\{40,80,160\}$ countries with an edge probability of $p=0.52$. Because every realized unordered pair is observed across all three categories, the design intrinsically dictates $L_d\equiv3$, $\mathcal{L}=3$, and $\mathcal{J}_N=3N$. Country attributes $\varphi_{g\ell}$ are drawn from independent $\mathrm{U}[-1,1]$ distributions, formally defining the node-borne economic size as $w_{d\ell}=\varphi_{g\ell}+\varphi_{h\ell}$. The disturbances are specified as $u_{d\ell}=\sigma_\ell(a_{g\ell}+a_{h\ell})+\varepsilon_{d\ell}$, where $a_{g\ell}=(\eta_g+\nu_{g\ell})/2$ and $\varepsilon_{d\ell}=(\xi_d+\zeta_{d\ell})/2$. These independent $\mathrm{U}[-1,1]$ primitives explicitly enforce the bounded components condition of Assumption (A2) and establish cross-category correlations of $\operatorname{Corr}(a_{g\ell},a_{g\ell'}) = \operatorname{Corr}(\varepsilon_{d\ell},\varepsilon_{d\ell'}) = 1/2$ for $\ell\ne\ell'$. These correlations directly populate the off-diagonal blocks of $\Omega_N$. The simulation executes 5,000 replications per configuration.

Heterogeneous convergence rates are induced by systematically varying the per-category scale of the country effects: $\sigma_\ell=G^{-\gamma_\ell/2}$ with $\gamma=(0,s/2,s)$ evaluated across the rate spreads $s\in\{0,1/4,0.45\}$. Because the covariate $w_{d\ell}$ is inherently node-borne, the unit aggregate $\sum_{h}w_{\{g,h\}\ell}$ operates at order $Gp$, yielding a category-$\ell$ score variance of order $G^{3-\gamma_\ell}p^{2}$. Consequently, the condition number scales as $\kappa(\Omega_N)\asymp G^{s}$, and the accumulation ratio follows $\delta_N\asymp G^{2s-1}\to 0$. While these asymptotic orders and the corresponding $\delta_N$ column in Table \ref{tab:mc_ml} satisfy the limit required by Assumption (A4) across all nine configurations, they do not put any scalar normalization $(NG^{r})^{-1}\Omega_N$ on course for a positive-definite limit when $s>0$. Accordingly, the condition number reported in Table \ref{tab:mc_ml} strictly increases with $s$. Interpreted through this framework, the design satisfies the condition imposed by Assumption 2.6 in \citet[p.~675]{tabordmeehan2019inference} only in the $s=0$ column, deliberately operating outside of it in the other two scenarios. These columns are intentionally constructed to straddle the assumption rather than to formally satisfy it. The configurations therefore exercise statistical inference precisely when the two contrasts accumulate at different asymptotic orders ($G^{3}$ and $G^{3-s/2}$). This operationalizes the theoretical discussion of Assumption (A4) in Section \ref{sec:assumptions}, demonstrating that valid inference is maintained even when joint restrictions accumulate dependence at genuinely heterogeneous rates.

The triangular array structure, in which $\sigma_\ell$ drifts with $G$, is a mathematical necessity in this context. On a dense support, every direction of $\Omega_N$ is confined strictly within the polynomial band between $G^{5/2}$ and $G^{3}$; the lower bound is dictated by Assumption (A4), while the upper bound is imposed by the cap $\lambda_{\max}(\Omega_N)=O(N\mathcal{M}^{H})$ derived from Assumption (A2) via Lemma \ref{lem:struct}(d). Establishing two distinct polynomial orders within this constrained band requires the data-generating process to drift with $G$, a feature that Theorem \ref{thm:main} structurally accommodates via its triangular array formulation. Furthermore, the fact that the covariates are node-borne ensures that every direction of $\Omega_N$ remains within this prescribed band. Conversely, a covariate exhibiting strictly idiosyncratic pair-level variation yields a unit aggregate of order $(Gp)^{1/2}$, causing $\Omega_N$ to accumulate at order $N$ along its contributed direction. On a dense support, this accumulation falls strictly below the floor imposed by Assumption (A4). As detailed in Illustration 2 of \ref{sec:smb}, this dynamic is a fundamental property of the specific covariate rather than a general characteristic of bilateral covariates as a class.  

Table \ref{tab:mc_ml} reports configuration diagnostics computed directly from $\Omega_N$, alongside Monte Carlo standard deviations for each category's slope and the raw rejection frequencies for nominal $5\%$ tests of the null hypothesis $H_0:\beta_{1,1}=\beta_{2,1}=\beta_{3,1}$ ($q=2$). The performance of $W_N$ is evaluated against an i.i.d. benchmark and a pair-level clustering estimator. The inverted variance matrices maintained strict positive definiteness across all $45{,}000$ replications; thus, the convention outlined in Section \ref{sec:assumptions} of setting singular Wald statistics to zero was never exercised during this simulation exercise.

\begin{table}[tp]
\centering
\spacingset{1}
\setlength{\tabcolsep}{4pt}
\footnotesize
\caption{Size of the test of $H_0:\beta_{1,1}=\beta_{2,1}=\beta_{3,1}$ across nine
multilayer configurations, in which the three layers' coefficients converge at
genuinely different rates.}
\label{tab:mc_ml}
\begin{tabular}{@{}rrrrrr rrr rrrr@{}}
\toprule
 & & & & & & \multicolumn{3}{c}{$\mathrm{SD}(\widehat\beta_{\ell,1})$}
 & \multicolumn{4}{c}{Size} \\
\cmidrule(lr){7-9}\cmidrule(lr){10-13}
$G$ & $s$ & $N$ & $\kappa(\Omega_N)$ & $\delta_N$ & $\mathcal{J}_N/\lambda_{\min}$
 & $\ell=1$ & $\ell=2$ & $\ell=3$
 & i.i.d. & Dyad & $W_N$ & $W^{J}_N$ \\
\midrule
 40 & 0.00 &   1{,}216 &  8.93 &   1{,}461 &  5.44 & 0.1204 & 0.1223 & 0.1178 & 0.672 & 0.686 & 0.127 & 0.053 \\
 40 & 0.25 &   1{,}219 & 12.27 &   5{,}468 & 10.52 & 0.1181 & 0.0956 & 0.0777 & 0.631 & 0.643 & 0.125 & 0.047 \\
 40 & 0.45 &   1{,}216 & 19.60 &  18{,}497 & 19.45 & 0.1197 & 0.0833 & 0.0577 & 0.593 & 0.612 & 0.144 & 0.048 \\
\addlinespace
 80 & 0.00 &   4{,}937 &  7.57 &       484 &  2.38 & 0.0805 & 0.0813 & 0.0825 & 0.818 & 0.820 & 0.078 & 0.046 \\
 80 & 0.25 &   4{,}934 & 12.74 &   2{,}979 &  5.89 & 0.0826 & 0.0624 & 0.0484 & 0.777 & 0.781 & 0.084 & 0.053 \\
 80 & 0.45 &   4{,}930 & 24.20 &  14{,}228 & 12.91 & 0.0809 & 0.0511 & 0.0328 & 0.734 & 0.746 & 0.081 & 0.045 \\
\addlinespace
160 & 0.00 &  19{,}845 &  6.91 &       181 &  1.10 & 0.0563 & 0.0571 & 0.0568 & 0.894 & 0.896 & 0.063 & 0.050 \\
160 & 0.25 &  19{,}832 & 14.16 &   1{,}794 &  3.47 & 0.0573 & 0.0412 & 0.0301 & 0.875 & 0.876 & 0.070 & 0.051 \\
160 & 0.45 &  19{,}835 & 31.88 &  12{,}053 &  8.96 & 0.0567 & 0.0325 & 0.0191 & 0.849 & 0.846 & 0.063 & 0.048 \\
\midrule
\multicolumn{6}{@{}l}{\emph{Limiting value}} & & & & 0.050 & 0.050 & 0.050 & 0.050 \\
\bottomrule
\end{tabular}

\vspace{4pt}
\begin{minipage}{\textwidth}\scriptsize
\textsuperscript{a}\,$G$ is the number of countries and $s$ the rate spread, the layers
carrying node-effect scales $\sigma_\ell=G^{-\gamma_\ell/2}$ with
$\gamma=(0,s/2,s)$; every unordered pair is present with probability $p=0.52$ and is
observed in all $L=3$ layers, and $N$ is the mean number of observations across the
5{,}000 replications. Further,
$\kappa(\Omega_N)=\lambda_{\max}(\Omega_N)/\lambda_{\min}(\Omega_N)$,
$\delta_N=N(\mathcal{M}^{H})^{3}/\lambda_{\min}(\Omega_N)^{2}$ is the quantity
Assumption~(A4) sends to zero, and $\mathcal{J}_N/\lambda_{\min}$ abbreviates
$\mathcal{J}_N/\lambda_{\min}(\Omega_N)$, the quantity Assumption~(A6) sends to zero, all
three computed exactly from $\Omega_N$ on the realized configuration and averaged over
draws.\\
\textsuperscript{b}\,$\mathrm{SD}(\widehat\beta_{\ell,1})$ is the Monte Carlo standard
deviation of the layer-$\ell$ slope across replications.\\
\textsuperscript{c}\,Entries in the last four columns are raw rejection frequencies of
nominal $5\%$ tests of $H_0:\beta_{1,1}=\beta_{2,1}=\beta_{3,1}$, so that $q=L-1=2$;
i.i.d.\ is the Wald test built from $\widehat\sigma^{2}(X'X)^{-1}$ and Dyad the one built
by clustering on the country pair, that is
$\sum_{d}\big(\sum_{n:\,\psi(n)=d}\mathbf{x}_n\widehat u_n\big)
\big(\sum_{n:\,\psi(n)=d}\mathbf{x}_n\widehat u_n\big)'$, neither carrying a
finite-sample cluster correction; the Limiting value row is the nominal level.
\end{minipage}
\end{table}

The simulation results yield four conclusions. First, the data-generating process successfully separates the convergence rates. At $G=160$, the three category slopes exhibit identical standard deviations for $s=0$, but diverge sharply for $s=0.45$ (at $0.0567$, $0.0325$, and $0.0191$). The condition number $\kappa(\Omega_N)$ separates correspondingly, rising from $6.91$ at $s=0$ to $31.88$ at $s=0.45$.  

Second, the size of the rate-agnostic Wald statistic $W_N$ is strictly governed by $G$, remaining invariant to the rate spread $s$. At $G=160$, the empirical size is bounded between $0.063$ and $0.070$ regardless of whether $\kappa(\Omega_N)$ diverges. The identical restriction matrix, variance estimator, and critical value maintain the correct asymptotic size in the absolute absence of a scalar normalization.

Third, the jackknife statistic $W^{J}_N$ strictly limits over-rejection, yielding a size between $0.045$ and $0.053$ across all configurations. This aligns with Proposition \ref{prop:onesided}, which bounds a limiting eigenvalue but does not formally order the two statistics within any given finite sample.

Finally, the other two commonly-used alternatives fail systematically, rejecting a true null hypothesis between $59\%$ and $90\%$ of the time, with performance deteriorating as $G$ grows. Clustering solely on the country pair discards the covariances between pairs that share a common node, ignoring the dominant source of dependence, which operates at order $G^{3}p^{2}$ on a multilayer support. We note that while $\delta_N$ correctly approaches zero in $G$, its magnitude remains large at these specific sample sizes, ranging from $181$ to $12{,}053$ at $G=160$. The size of $W_N$ is nevertheless stable across the row, remaining between $0.063$ and $0.070$ at $G=160$.

\subsection{Empirical Application}
\label{sec:empirical}

The empirical application estimates equation \eqref{eq:gravity} using the \citet{santossilva2006log} cross-section of bilateral merchandise trade among $G=136$ countries in 1990. Out of $18{,}360$ possible ordered pairs, $N=9{,}613$ ($52.4\%$) record positive trade and thus constitute the estimation sample. The estimand is strictly defined as the linear projection of $\ln T_{ij}$ onto the covariates within this positive-trade subsample. Because a projection residual is orthogonal to the regressors in the aggregate by construction, we impose Assumption (A3) directly at the level of the individual observation. \citet{santossilva2006log} argue that least squares on log trade fails to identify the elasticities of $\mathbb{E}[T_{ij}\mid\mathbf{x}]$ under heteroskedasticity; our approach does not dispute this. No coefficient estimated here is interpreted as an elasticity or read causally, as this exercise focuses strictly on statistical \emph{inference} rather than model specification. Consequently, zero-trade pairs are dropped as a fundamental part of the estimand's definition rather than as an empirical recommendation. This omission renders the configuration incomplete, resulting in genuinely variable degree counts: $\mathcal{M}^{H}=270$ against $\mathcal{M}^{L}=40$, indicating a network that is dense at the top and sparse at the bottom. Furthermore, the multiplicity count is $\mathcal{J}_N=17{,}775$, compared to the $2N=19{,}226$ that a fully reciprocated support would yield; specifically, $4{,}081$ of the $5{,}532$ realized dyads are observed in both orientations, while the remaining $1{,}451$ are observed in only one.

The exporter attributes, $\mathbf{z}_i$, comprise log GDP (\texttt{ly}), log GDP per capita (\texttt{lyp}), log remoteness (\texttt{lremot}), and a landlocked dummy (\texttt{landl}). These are duplicated symmetrically for the importer attributes, $\mathbf{z}_j$. The symmetric pair attributes, $\mathbf{s}_{\{i,j\}}$, include log distance (\texttt{ldist}), contiguity (\texttt{border}), common language (\texttt{comlang}), colonial ties (\texttt{colony}), a common free-trade-agreement dummy (\texttt{comfrt\_wto}), an open-economy dummy (\texttt{open\_wto}), and the constant, establishing $K=15$. We test the joint hypothesis of push-pull symmetry, formalized as $H_0:\boldsymbol{\beta}_O=\boldsymbol{\beta}_D$ ($q=4$). As detailed in Illustration 1 in \ref{sec:smb}, this hypothesis is mathematically equivalent to the claim that the regression function forgets the orientation, exactly as the mapping $\psi$ does.

Table \ref{tab:emp_coef} reports the estimated coefficients alongside their i.i.d., dyadic-robust, and jackknife standard errors. For every coefficient, the i.i.d. standard errors are strictly smaller than both dyadic estimates: the dyadic-robust standard errors are $1.3$ to $4.2$ times larger than the i.i.d. values, and the jackknife standard errors are $1.7$ to $4.7$ times larger. Table \ref{tab:emp_wald} reports the corresponding Wald tests. While the joint null hypothesis is decisively rejected across all estimators, accounting for dyadic dependence substantially moderates the magnitude of this rejection: $W^{J}_N=46.12$ compared to $\chi^2_{4,0.95}=9.49$, whereas the i.i.d. statistic is massively inflated to $295.53$. Notably, for the total GDP attribute, the single-restriction conclusion is entirely reversed depending on the variance estimator. The i.i.d. test rejects equality at $21.02$, whereas the jackknife test fails to reject at $2.96$ against a critical value of $\chi^2_{1,0.95}=3.84$. Per-capita income, remoteness, and landlockedness remain individually asymmetric regardless of the estimator used.

The jackknife statistic $W^{J}_N$ lies strictly below $W(\widehat{V}_N)$ in every row of Table \ref{tab:emp_wald}, an observation consistent with the upward one-sided error established in Proposition \ref{prop:onesided} in the main text. However, because that proposition is asymptotic and stated relative to $\mathcal{V}_N$, it does not theoretically imply $\widehat{V}^{J}-\widehat{V}_N\succeq0$ in any given finite sample, nor does it guarantee a consistent ordering of the two statistics restriction by restriction. Note that the ordering observed here is strictly an empirical finding. Both $R_N'\widehat{V}_NR_N$ and $R_N'\widehat{V}^{J}R_N$ remain strictly positive definite for the $G=136$ sample across every tested restriction. Thus, the convention from Section \ref{sec:assumptions} of setting singular Wald statistics to zero is never exercised in this full sample. Conversely, across twelve random subsamples of $G=40$ countries, $\widehat{\Omega}_N$ fails to be positive definite in every instance. It similarly fails in eleven out of twelve subsamples at $G=60$, only achieving consistent positive definiteness for $G\ge80$. In these smaller samples, the jackknife estimator proves essential by supplying positive semidefiniteness by construction. Ultimately, this exercise formally evaluates, via dyadic-robust joint inference, the push-pull coefficient comparison that \citet{santossilva2006log} drew informally, demonstrating that the choice of variance estimator is highly consequential even when the headline conclusion remains intact.

\begin{table}[tp]
\centering
\spacingset{1}
\caption{Least-squares estimates of equation \eqref{eq:gravity} on the
\citet{santossilva2006log} sample.}
\label{tab:emp_coef}
\begin{tabular}{lrrrrr}
\toprule
Coefficient & Estimate & SE (i.i.d.) & SE ($\widehat{V}_N$) & SE ($\widehat{V}^{J}$) & Ratio \\
\midrule
\texttt{ly} (exporter)      & $0.2073$  & $0.0166$ & $0.0696$ & $0.0785$ & $4.72$ \\
\texttt{lyp} (exporter)     & $0.9378$  & $0.0116$ & $0.0419$ & $0.0474$ & $4.08$ \\
\texttt{lremot} (exporter)  & $0.4671$  & $0.0778$ & $0.2388$ & $0.2699$ & $3.47$ \\
\texttt{landl} (exporter)   & $-0.0620$ & $0.0646$ & $0.1548$ & $0.1791$ & $2.77$ \\
\texttt{ly} (importer)      & $0.1061$  & $0.0167$ & $0.0575$ & $0.0650$ & $3.89$ \\
\texttt{lyp} (importer)     & $0.7978$  & $0.0111$ & $0.0418$ & $0.0470$ & $4.24$ \\
\texttt{lremot} (importer)  & $-0.2050$ & $0.0808$ & $0.2689$ & $0.2994$ & $3.71$ \\
\texttt{landl} (importer)   & $-0.6645$ & $0.0631$ & $0.1224$ & $0.1446$ & $2.29$ \\
\texttt{ldist}              & $-1.1660$ & $0.0339$ & $0.0792$ & $0.0912$ & $2.69$ \\
\texttt{border}             & $0.3140$  & $0.1425$ & $0.1869$ & $0.2481$ & $1.74$ \\
\texttt{comlang}            & $0.6780$  & $0.0640$ & $0.1818$ & $0.2162$ & $3.38$ \\
\texttt{colony}             & $0.3968$  & $0.0681$ & $0.1941$ & $0.2297$ & $3.37$ \\
\texttt{comfrt\_wto}         & $0.4908$  & $0.1053$ & $0.2990$ & $0.3411$ & $3.24$ \\
\texttt{open\_wto}           & $-0.1696$ & $0.0490$ & $0.1843$ & $0.2085$ & $4.25$ \\
Constant                    & $-28.4920$& $1.0880$ & $4.0123$ & $4.4828$ & $4.12$ \\
\bottomrule
\end{tabular}

\vspace{4pt}
\begin{minipage}{\textwidth}\footnotesize
\textsuperscript{a}\,Estimate is the least-squares point estimate; SE (i.i.d.) is the
classical least-squares standard error; SE ($\widehat{V}_N$) uses the dyadic-robust
estimator of equation \eqref{eq:vhat} in the main text; SE ($\widehat{V}^{J}$) uses the
delete-one-unit jackknife of equation \eqref{eq:vjack} in the main text; Ratio is SE
($\widehat{V}^{J}$) over SE (i.i.d.).
\end{minipage}
\end{table}

\begin{table}[tp]
\centering
\spacingset{1}
\caption{Wald tests of push--pull symmetry, $H_0:\boldsymbol{\beta}_O=
\boldsymbol{\beta}_D$ and its single restrictions.}
\label{tab:emp_wald}
\begin{tabular}{lrrrrrr}
\toprule
Restriction & $q$ & $W$ (i.i.d.) & $W$ ($\widehat{V}_N$) & $W^{J}$ &
$\chi^2_{q,0.95}$ & $p$ ($W^{J}$) \\
\midrule
Joint: all four attributes            & 4 & $295.53$ & $61.15$ & $46.12$ & $9.49$ & $2.3\times10^{-9}$ \\
\texttt{ly}: exporter $=$ importer     & 1 & $21.02$  & $3.80$  & $2.96$  & $3.84$ & $0.085$ \\
\texttt{lyp}: exporter $=$ importer    & 1 & $91.35$  & $15.52$ & $12.03$ & $3.84$ & $5.2\times10^{-4}$ \\
\texttt{lremot}: exporter $=$ importer & 1 & $41.39$  & $7.13$  & $5.69$  & $3.84$ & $0.017$ \\
\texttt{landl}: exporter $=$ importer  & 1 & $44.34$  & $23.86$ & $16.06$ & $3.84$ & $6.1\times10^{-5}$ \\
\bottomrule
\end{tabular}

\vspace{4pt}
\begin{minipage}{\textwidth}\footnotesize
\textsuperscript{a}\,$q$ is the number of restrictions, with $R_N=(I_4;-I_4;0_{7\times4})$
for the joint test and the corresponding single rows for the others; $W$ (i.i.d.), $W$
($\widehat{V}_N$) and $W^{J}$ are the Wald statistics built from the classical,
dyadic-robust, and jackknife variance estimators respectively; $\chi^2_{q,0.95}$ is the
$5\%$ critical value; $p$ ($W^{J}$) is the $\chi^2_q$ $p$-value of the jackknife
statistic.
\end{minipage}
\end{table}

\spacingset{1}
\putbib[references]
\end{bibunit}

\end{document}